\documentclass[12pt]{article}
\usepackage{latexsym}
\usepackage{graphicx}
\usepackage{graphics}
\usepackage{amsmath,amssymb,amsfonts}
\usepackage[framemethod=tikz]{mdframed}
\usepackage{algorithm}
\usepackage[noend]{algpseudocode}
\usepackage{epstopdf}
\usepackage{makeidx}
\usepackage{fullpage}
\usepackage{hyperref}
\usepackage{color}
\usepackage{wrapfig}
\usepackage{xspace}
\usepackage{setspace}
\usepackage{algorithm}

\begin{document}
\newtheorem{theorem}{Theorem}[section]
\newtheorem{corollary}[theorem]{Corollary}
\newtheorem{lemma}[theorem]{Lemma}
\newtheorem{observation}[theorem]{Observation}
\newtheorem{proposition}[theorem]{Proposition}
\newtheorem{definition}[theorem]{Definition}
\newtheorem{claim}[theorem]{Claim}
\newtheorem{fact}[theorem]{Fact}
\newtheorem{assumption}[theorem]{Assumption}

\newcommand{\qed}{\rule{7pt}{7pt}}
\newcommand{\dis}{\mathop{\mbox{\rm d}}\nolimits}
\newcommand{\per}{\mathop{\mbox{\rm per}}\nolimits}
\newcommand{\area}{\mathop{\mbox{\rm area}}\nolimits}
\newcommand{\cw}{\mathop{\rm cw}\nolimits}
\newcommand{\ccw}{\mathop{\rm ccw}\nolimits}
\newcommand{\DIST}{\mathop{\mbox{\rm DIST}}\nolimits}
\newcommand{\OP}{\mathop{\mbox{\it OP}}\nolimits}
\newcommand{\OPprime}{\mathop{\mbox{\it OP}^{\,\prime}}\nolimits}
\newcommand{\ihat}{\hat{\imath}}
\newcommand{\jhat}{\hat{\jmath}}
\newcommand{\abs}[1]{\mathify{\left| #1 \right|}}

\newenvironment{proof}{\noindent{\bf Proof}\hspace*{1em}}{\qed\bigskip}
\newenvironment{proof-sketch}{\noindent{\bf Sketch of Proof}\hspace*{1em}}{\qed\bigskip}
\newenvironment{proof-idea}{\noindent{\bf Proof Idea}\hspace*{1em}}{\qed\bigskip}
\newenvironment{proof-of-lemma}[1]{\noindent{\bf Proof of Lemma #1}\hspace*{1em}}{\qed\bigskip}
\newenvironment{proof-attempt}{\noindent{\bf Proof Attempt}\hspace*{1em}}{\qed\bigskip}
\newenvironment{proofof}[1]{\noindent{\bf Proof}
of #1:\hspace*{1em}}{\qed\bigskip}
\newenvironment{remark}{\noindent{\bf Remark}\hspace*{1em}}{\bigskip}


\newcommand{\FOR}{{\bf for}}
\newcommand{\TO}{{\bf to}}
\newcommand{\DO}{{\bf do}}
\newcommand{\REPEAT}{{\bf repeat}}
\newcommand{\UNTIL}{{\bf until}}
\newcommand{\WHILE}{{\bf while}}
\newcommand{\AND}{{\bf and}}
\newcommand{\IF}{{\bf if}}
\newcommand{\THEN}{{\bf then}}
\newcommand{\ELSE}{{\bf else}}

\makeatletter
\def\fnum@figure{{\bf Figure \thefigure}}
\def\fnum@table{{\bf Table \thetable}}
\long\def\@mycaption#1[#2]#3{\addcontentsline{\csname
  ext@#1\endcsname}{#1}{\protect\numberline{\csname 
  the#1\endcsname}{\ignorespaces #2}}\par
  \begingroup
    \@parboxrestore
    \small
    \@makecaption{\csname fnum@#1\endcsname}{\ignorespaces #3}\par
  \endgroup}
\def\mycaption{\refstepcounter\@captype \@dblarg{\@mycaption\@captype}}
\makeatother

\newcommand{\figcaption}[1]{\mycaption[]{#1}}
\newcommand{\tabcaption}[1]{\mycaption[]{#1}}
\newcommand{\head}[1]{\chapter[Lecture \##1]{}}
\newcommand{\mathify}[1]{\ifmmode{#1}\else\mbox{$#1$}\fi}
\newcommand{\bigO}O
\newcommand{\set}[1]{\mathify{\left\{ #1 \right\}}}
\def\half{\frac{1}{2}}


\newcommand{\enc}{{\sf Enc}}
\newcommand{\dec}{{\sf Dec}}
\newcommand{\E}{{\rm Exp}}
\newcommand{\Var}{{\rm Var}}
\newcommand{\Z}{{\mathbb Z}}
\newcommand{\F}{{\mathbb F}}
\newcommand{\integers}{{\mathbb Z}^{\geq 0}}
\newcommand{\R}{{\mathbb R}}
\newcommand{\Q}{{\cal Q}}
\newcommand{\eqdef}{{\stackrel{\rm def}{=}}}
\newcommand{\from}{{\leftarrow}}
\newcommand{\vol}{{\rm Vol}}
\newcommand{\poly}{{\rm poly}}
\newcommand{\ip}[1]{{\langle #1 \rangle}}
\newcommand{\wt}{{\rm wt}}
\renewcommand{\vec}[1]{{\mathbf #1}}
\newcommand{\mspan}{{\rm span}}
\newcommand{\rs}{{\rm RS}}
\newcommand{\RM}{{\rm RM}}
\newcommand{\Had}{{\rm Had}}
\newcommand{\calc}{{\cal C}}

\newcommand{\fig}[4]{
        \begin{figure}
        \setlength{\epsfysize}{#2}
        \vspace{3mm}
        \centerline{\epsfbox{#4}}
        \caption{#3} \label{#1}
        \end{figure}
        }

\newcommand{\ord}{{\rm ord}}

\providecommand{\norm}[1]{\lVert #1 \rVert}
\newcommand{\embed}{{\rm Embed}}
\newcommand{\qembed}{\mbox{$q$-Embed}}
\newcommand{\calh}{{\cal H}}
\newcommand{\lp}{{\rm LP}}

\newcommand{\sspace}{\baselineskip 14pt}
\newcommand{\dspace}{\baselineskip 24pt}
\newcommand{\ASG}{\leftarrow}
\newcommand{\TN}{{\bf then}}
\newcommand{\EL}{{\bf else}}
\newcommand{\EI}{{\bf endif}}
\newcommand{\RE}{{\bf repeat}}
\newcommand{\UN}{{\bf until}}
\newcommand{\ER}{{\bf endrepeat}}
\newcommand{\WH}{{\bf while}}
\newcommand{\EW}{{\bf endwhile}}
\newcommand{\FO}{{\bf for}}
\newcommand{\BY}{{\bf by}}
\newcommand{\EF}{{\bf endfor}}
\newcommand{\CA}{{\bf case}}
\newcommand{\EC}{{\bf endcase}}
\newcommand{\TR}{{\bf true}}
\newcommand{\FA}{{\bf false}}
\newcommand{\T}{\hspace*{.3in}}

\def\ignore#1{\relax}
\providecommand{\email}[1]{\href{mailto:#1}{\nolinkurl{#1}\xspace}}
\pagestyle{empty}
\title{Optimal Parametric Search for Path and Tree Partitioning}

\author{Greg N. Frederickson\thanks{Supported in part by the National Science Foundation under
grants CCR-86202271 and CCR-9001241,
and by the Office of Naval Research under contract
N00014-86-K-0689.} \\
\and Samson Zhou\thanks{Research supported by NSF CCF-1649515.
Email: \email{samsonzhou@gmail.com}
}}
\date{\today}
\dspace
\maketitle
\pagestyle{plain}
\noindent
{\bf Abstract.}
We present linear-time algorithms for partitioning a path or a tree with weights on the vertices by removing $k$ edges to maximize the minimum-weight component. We also use the same framework to partition a path with weight on the vertices, removing $k$ edges to minimize the  maximum-weight component. The algorithms use the parametric search paradigm, testing candidate values until an optimum is found while simultaneously reducing the running time needed for each test.
For path-partitioning, the algorithm employs a synthetic weighting scheme that results in a constant fraction reduction in running time after each test. 
For tree-partitioning, our dual-pronged strategy makes progress no matter what the specific structure of our tree is.
\vspace{.2in}\\
{\bf Key words and phrases.}
Adaptive algorithms, 
data structures,
max-min,
min-max,
parametric search,
partial order,
path partitioning,
prune-and-search,
sorted matrices,
tree partitioning.
\newpage
\pagenumbering{arabic}
\dspace

\section{Introduction}

Parametric search is a powerful technique for solving various
optimization problems
\cite{AP,C1,C2,CSSS,FJ1,FJ2,Gu,M,R,S,Z}.
To solve a given problem by parametric search, one must find a threshold value, namely a largest (or smallest) value
that passes a certain feasibility test.
In general, we must identify explicitly or implicitly a set of candidate values, and search the set, 
choosing one value at a time upon which to test feasibility.
As a result of the test, we discard from further consideration various values in the set.
Eventually, the search will converge on the optimal value.
 
In \cite{M}, Megiddo emphasized the role that
parallel algorithms can play in the implementation of such
a technique.
In \cite{C2}, Cole improved the time bounds for a number of problems
addressed in \cite{M},
abstracting two techniques upon which parametric search can be based:
sorting and performing independent binary searches.
Despite the clever ideas elaborated in these papers,
none of the algorithms presented in \cite{M, C2} 
is known to be optimal.
In fact, rarely have algorithms for parametric search problems
been shown to be optimal.
(For one that is optimal, see \cite{CSSS}.)
In at least some cases, the time required for feasibility testing
matches known lower bounds for the parametric search problem.
But in the worst case, $\Omega (\log n)$ values must be tested for feasibility, where $n$ is the size of the input.
Is this extra factor of at least $\log n$ necessary?
For several parametric search problems on paths and trees,
we show that a polylogarithmic penalty in the running time can be avoided,
and give linear-time (and hence optimal) algorithms for these problems.

We consider the max-min tree $k$-partitioning problem \cite{PS}.
Let $T$ be a tree with $n$ vertices and a nonnegative weight associated with each vertex. 
Let $k < n$ be a given positive integer.
The problem is to delete $k$ edges in the tree so as to maximize the weight of the lightest of the resulting subtrees.
Perl and Schach introduced an algorithm that runs in $O(k^2 rd(T)+kn)$ time, where $rd(T)$ is the radius of the tree \cite{PS}.
Megiddo and Cole first considered the special case in which the tree is a path, and gave $O(n (\log n)^2)$-time and $O(n \log n)$-time algorithms, resp., for this case.\footnote{All logarithms are to the base 2.}
(Megiddo noted that an $O(n \log n)$-time algorithm is possible, using ideas from \cite{FJ1}.)
For the problem on a tree, Megiddo presented an $O(n( \log n)^3)$-time algorithm \cite{M}, and Cole presented an $O(n( \log n)^2)$-time algorithm \cite{C2}. 
The algorithm we present here runs in $O(n)$ time, which is clearly optimal.

A closely related problem is the min-max path $k$-partitioning problem \cite{BPS},  for which we must delete $k$ edges in the path so as to minimize the weight of the heaviest of the resulting subpaths. 
Our techniques yield linear-time algorithms for both versions of path-partitioning  and for max-min tree $k$-partitioning problem. Our approach is less convoluted than that in \cite{FTR, F1}, and thus its analysis should be more amenable to independent confirmation and publication. 
We believe that this paper provides for the full acknowledgement of linear time  algorithms for these particular problems.

Our results are a contribution to parametric search in several ways
beyond merely producing optimal algorithms for path and tree partitioning.  
In contrast to \cite{F1}, we introduce synthetic weights on candidate values 
or groups of candidate values from a certain implicit matrix. 
We use weighted selection to resolve the values representing many shorter paths before the values for longer paths, enabling future feasibility tests to run faster. 
We also use unweighted selection to reduce the size of the set of candidate values so that the time for selecting 
test values does not become an obstacle to achieving linear time. 

Our parametric search on trees reduces feasibility test time across subpaths, while simultaneously also pruning paths that
have no paths that branch off of them. To demonstrate progress, we identify an effective measure of current problem size. 
Surprisingly, the number of vertices remaining in the tree seems not to be so helpful, as is also the case for the number of candidate values. 
Instead, we show that the time to perform a feasibility test actually is effective. 
Our analysis for tree partitioning captures the total progress from the beginning of the algorithm, not necessarily between consecutive rounds, as was the case for the path-partitioning problem. 
Thus it seems necessary to have a dual-pronged strategy to simultaneously compress the search structure of paths and also prune paths that end with a leaf.

We organize our paper as follows. 
In Section 2 we discuss features of parametric search, and lay the foundation for the optimal algorithms that we describe in subsequent sections.
In Section 3 we present our approach for solving both
the min-max and the max-min partitioning problem on a path. 
In Section 4 we build on the results in Section 3 and
present our approach for optimally solving the max-min partitioning problem on a tree.
\bigskip

\section{Paradigms for parametric search}
\label{sec:prelims}

In this section we discuss features of parametric search,
and lay the foundation for the optimal algorithms
that we give in subsequent sections.
We first review straightforward feasibility tests
to test the feasibility of a search value in a tree.
We then review how to represent all possible search values
of a path within a ``sorted matrix''. 
We next present a general approach for search that
uses values from a collection of sorted matrices.
Finally, we describe straightforward approaches
for the path and the tree that are as good as any
algorithms in \cite{M,C2}.

We first describe straightforward feasibility tests for cutting edges in a tree so as to maximize the minimum weight of any resulting component (max-min problem).
The feasibility test takes a test value $\lambda$, and determines if at least $k$ cuts can be placed in the tree
such that no component has weight less than $\lambda$. 
We take $\lambda^*$ to be the largest value that passes the test.

A straightforward test for the max-min problem in a tree is given in \cite{M},
and the feasibility test for the min-max problem is similar \cite{KM}. 
We focus first on max-min, using an algorithm {\it FTEST0}, which starts 
by rooting the tree at a vertex of degree 1
and initializing the number of cuts, $numcut$, to $0$.
It next calls {\it explore}$({\it root})$ to explore the tree starting at the root.
When the exploration is complete,
if $numcut > k$,
then $\lambda$ is a lower bound on $\lambda^*$,
and otherwise $\lambda$ is an upper bound on $\lambda^*$.
The procedure {\it explore} is:

\bigskip
\sspace
\noindent
{\bf proc} {\it explore} $({\bf vertex} ~v)$\vspace{.05in}\\
$\T $ {\it accum\_wgt}$(v)$ $\ASG$ weight of $v$\vspace{.05in}\\
$\T $ $\FO$ each child $w$ of $v$ $\DO$ ${\it explore(w)}$ $\EF$\vspace{.05in}\\
$\T $ Adjust {\it accum\_wgt}$(v)$ and $numcut$.
 
\dspace
\bigskip
\noindent
For the max-min problem, we adjust {\it accum\_wgt}$(v)$ and $numcut$ by:

\bigskip
\sspace
\noindent
$\T $ $\FO$ each child $w$ of $v$ $\DO$ Add ${\it accum\_wgt}(w)$ to ${\it accum\_wgt}(v)$. $\EF$\vspace{.05in}\\
$\T $ $\IF$ {\it accum\_wgt}$(v) \geq \lambda$\vspace{.05in}\\ 
$\T $ $\TN$ $numcut \ASG numcut + 1$\vspace{.05in}\\
$\T $ $\T $ {\it accum\_wgt}$(v) \ASG 0$\vspace{.05in}\\
$\T $ $\EI$
 
\dspace
\bigskip

In all but one case, we cut an edge in the tree whenever we increment $numcut$ by 1. 
For the max-min problem, we generally cut the edge from the current vertex $v$ to its parent except for 
the last increment of $numcut$ in {\it FTEST0}.
This is because either $v$ will be the root and thus does not have a parent,
or the fragment of the tree above $v$ will have total weight less than $\lambda$.
By not adjusting $numcut$ to reflect the actual number of cuts in this case,
we are able to state the threshold test for both versions of {\it FTEST0}
in precisely the same form. 
The feasibility test takes constant time per vertex, and thus uses $O(n)$ time.

In a parametric search problem,
the search can be restricted to considering values from a finite set of values.
The desired value is the largest (in the case of a max-min problem)
or the smallest (in the case of a min-max problem)
that passes the feasibility test.
Let each value in the finite set be called a {\it candidate value}.
We next discuss a data structure that contains all candidate
values for problems on just a path.
This structure is based on ideas in \cite{FJ1} and \cite{FJ2}.
Let a matrix be called a {\it sorted matrix}
if for every row the values are in nondecreasing order,
and for every column the values are in nonincreasing order.
(Note that for notational convenience,
this definition varies slightly from the definition in \cite{FJ1}.)
Let the vertices on the path $P$ be indexed from 1 to $n$.

The set of values that we need to search is the set of sums of weights of vertices from $i$ to $j$, for all pairs $i\leq j$. 
We succinctly maintain this set of values, plus others, in a data structure, the \emph{succinct description}. 
For $i=0,1, \cdots ,n$, let $A_i$ be the sum of the weights of vertices 1 to $i$.
Note that for any pair $i,j$ with $i \leq j$, the sum of the weights from vertex $i$ to $j$ is $A_j - A_{i-1}$.
Let $X_1$ be the sequence of sums $A_1, A_2, \cdots ,A_n$, and let $X_2$ be the sequence of sums $A_0, A_1, \cdots ,A_{n-1}$.
Then sorted matrix $M(P)$ is the $n \times n$ Cartesian matrix $X_1-X_2$, where the $ij$-th entry is $A_j-A_{i-1}$.
In determining the above, we can use proper subtraction, in which, $a-b$ gives $\max \{a-b,0\}$.
Clearly, the values in any row of $M(P)$ are in nondecreasing order, and the values in any column of $M(P)$ are in nonincreasing order. 
Representing $M(P)$ explicitly would require $\Theta(n^2)$ time and space. 
Thus, our data structure succinctly represents $M(P)$ (and hence, $P$) in $O(n)$ space by the two vectors $X_1$ and $X_2$. 

In general, our algorithm also needs to inspect specific subpaths of $P$. 
However, repeatedly copying subvectors of $X_1$ and $X_2$ can take more than linear time in total. 
On the other hand, for a subpath $Q$ of $P$, the corresponding matrix $M(Q)$ is a submatrix of $M(P)$, which we can recover from the succinct representation of $P$. 
Thus, our algorithm succinctly represents $M(Q)$ by the start and end indices of $M(Q)$ within $M(P)$. 
In this way, we can generate the values of $M(Q)$ from the vectors $X_1$ and $X_2$ of $M(P)$, as well as the location of the submatrix as given by the succinct representation of $Q$. 
Therefore, our algorithm avoids needlessly recopying vectors and instead generates in $O(n)$ total time the succinct representations of all subpaths that it may inspect.

\begin{figure}[thb]
\begin{center}
\includegraphics{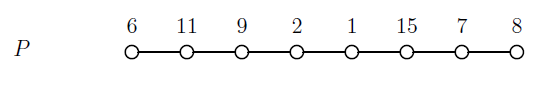}
\includegraphics{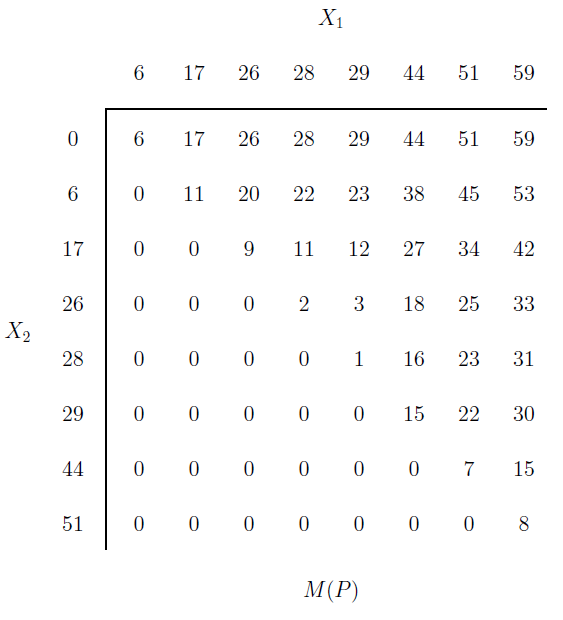}
\end{center}
\caption{\small Path $P$, vectors $X_1$ and $X_2$, and explicit illustration of matrix $M(P)$
\label{fig1}}
\end{figure}

As an example, we show a vertex-weighted path $P$ in Fig.~\ref{fig1},
and its associated matrix $M(P)$.
We list the sequence $X_1$ horizontally above $M(P)$,
and the sequence $X_2$ vertically to the left of $M(P)$,
in such a way that the $ij$-th element of $M(P)$
is beneath the $j$-th element of $X_1$
and to the right of the $i$-th element of $X_2$.

We next describe the general form {\it PARAM\_SEARCH} of all of our searching algorithms.
It is related to, and uses some of the ideas in, algorithms found in \cite{C2}, \cite{FJ1}, \cite{FJ2}, and \cite{M}.
By specifying the specific subroutines {\it INIT\_MAT}, {\it TEST\_VAL}, {\it UPDATE\_MAT},  we will be able to give four versions of {\it PARAM\_SEARCH}, namely {\it PATH0}, {\it TREE0}, {\it PATH1}, and {\it TREE1}.
In all four versions, {\it PARAM\_SEARCH} takes as its arguments an integer $k$ and either a vertex-weighted path or a tree.
It initializes searching bounds $\lambda_1$ and $\lambda_2$, where $\lambda_1 < \lambda_2$. 
The algorithm progressively narrows these bounds until they satisfy the following conditions by the end of the algorithm.
For a max-min problem, $\lambda_1$ is the largest value that is feasible,
and $\lambda_2$ is the smallest value that is not feasible.
(For a min-max problem, $\lambda_1$ is the largest value that is not feasible,
and $\lambda_2$ is the smallest value that is feasible.)

{\it PARAM\_SEARCH} will use {\it INIT\_MAT} to initialize $\cal M$,
a collection of succinctly represented square sorted matrices,
the union of whose values is the set of values that we consider.
{\it PARAM\_SEARCH} then performs a series of iterations.
On each iteration, it will use {\it TEST\_VAL} to identify and test a
small number of values $\lambda$
drawn from matrices in $\cal M$.
Each value $\lambda$ will be either the largest or the smallest element
in some matrix in $\cal M$.
As a result of the feasibility tests,
{\it UPDATE\_MAT} updates $\mathcal{M}$ by deleting 
certain matrices from $\cal M$,
dividing certain matrices into 4 submatrices,
and inserting certain matrices into $\cal M$.
Note that in initializing $\lambda_2$ to $\infty$,
we take $\infty$ to be any value greater than the total weight
of all vertices in the path or the tree.

\bigskip
\sspace
\noindent
{\bf Algorithm} {\it PARAM\_SEARCH}\vspace{.05in}\\
$\T $ $\lambda_1 \ASG 0$\\
$\T $ $\lambda_2 \ASG \infty$\vspace{.05in}\\
$\T $ {\it INIT\_MAT}\vspace{.05in}\\
$\T $ $\WH$ $\cal M$ is not empty $\DO$\vspace{.05in}\\
$\T \T $ {\it TEST\_VAL}\vspace{.05in}\\
$\T \T $ {\it UPDATE\_MAT}\vspace{.05in}\\
$\T $ $\EW$\\
$\T $ Output $\lambda_1$ and $\lambda_2$.\\
$\T $ /* For max-min, $\lambda^*$ will be the final $\lambda_1$ and for min-max, $\lambda^*$ will be the final $\lambda_2$. */
 
\dspace
\bigskip

In the remainder of this section we
describe simple algorithms {\it PATH0} and {\it TREE0}
for partitioning a path and a tree, resp.
These algorithms match the time of the algorithms in \cite{C2},
and set the stage for the improved algorithms
that we present in the next two sections, 
in which we introduce data structures that enable faster feasibility tests 
and strategies that prune the tree quickly. 
We first describe a simple approach to the max-min problem on a path $P$.
Below are the three routines 
{\it PATH0\_init\_mat},
{\it PATH0\_test\_val},
and {\it PATH0\_update\_mat}.
For a max-min problem,
if $\lambda > \lambda_1$ and $\lambda$ is feasible, then we reset $\lambda_1$ to $\lambda$.
Otherwise, if $\lambda < \lambda_2$ and $\lambda$ is not feasible,
then we reset $\lambda_2$ to $\lambda$. 
Thus at termination, $\lambda^*=\lambda_1$.
(For the min-max problem, we reset $\lambda_2$ to $\lambda$ if $\lambda$ is feasible,
and $\lambda_1$ to $\lambda$ if $\lambda$ is not feasible. 
Thus, at termination, $\lambda^*=\lambda_2$.)
On every iteration we split matrices of size greater than $1\times1$ into four smaller submatrices.
We assume that the dimension
of each sorted matrix is a power of 2.
If this is not the case,
then we pad out the matrix logically with zeroes.
In the following, the notation $(\lambda_1,\lambda_2)$
denotes the open interval of values between $\lambda_1$ and $\lambda_2$. 
We also use $R(P)$ to denote the representatives whose values are within the interval $(\lambda_1,\lambda_2)$.

\bigskip
\sspace
\noindent
{\it PATH0\_init\_mat}:\vspace{.05in}\\
$\T $ Implicitly split sorted matrix $M(P)$ for the path $P$ into four square submatrices. \\
$\T $ Initialize $\cal M$ to be the set containing these four square submatrices.

\bigskip
\noindent
{\it PATH0\_test\_val}:\vspace{.05in}\\
$\T $ $\IF$ each submatrix in $\cal M$ contains just 1 element \\
$\T $ $\TN$ Let $R$ be the multiset of values in the submatrices. \\
$\T $ $\EL$ Let $R$ be the multiset consisting of the smallest and the largest \\
$\T \T$ element from each matrix in ${\cal M}$. \\
$\T $ $\EI$ \\
$\T $ $\FO$ two times $\DO$: \\
$\T \T$ Let $R'$ be the subset of $R$ that contains only values in the interval $(\lambda_1,\lambda_2)$. \\
$\T \T$ $\IF$ $R'$ is not empty $\TN$ \\
$\T \T \T$ Select the median element $\lambda$ in $R'$. \\
$\T \T \T$ $\IF$ {\it FTEST0}$(P,k,\lambda ) = ``lower''$ (\emph{i.e.}, $k<numcuts$) \\
$\T \T \T \T$ $\TN$ $\lambda_1 \ASG \lambda$ \\
$\T \T \T$ $\EL$ $\lambda_2 \ASG \lambda$ \\
$\T \T \T$ $\EI$ \\
$\T \T$ $\EI$ \\
$\T \EF$

\bigskip
\noindent
{\it PATH0\_update\_mat}:\vspace{.05in}\\
$\T$ Discard from ${\cal M}$ any matrix with no values in $(\lambda_1,\lambda_2)$. \\
$\T $ $\IF$ each submatrix in $\cal M$ contains more than 1 element \\
$\T $ $\TN$ Split each submatrix $M$ in $\cal M$ into four square submatrices, \\
$\T \T $ discarding any resulting submatrix with no values in $(\lambda_1,\lambda_2)$. \\
$\T $ $\EI$ 
 
\dspace
\bigskip

The following lemma is similar in spirit to Lemma 5 in \cite{FJ1}
and Theorem 2 in \cite{FJ2}.

\bigskip

\begin{lemma}
\label{lem:2:1}
Let $P$ be a path of $n>2$ vertices.
The number of iterations needed by {\it PATH0} is $O(\log n)$,
and the total time of {\it PATH0} exclusive of feasibility tests is $O(n)$.
\end{lemma}
\begin{proof}
We call the multiset of smallest and largest values from each submatrix in $\cal M$ the {\it representatives} of $\cal M$,
and we call the subset of representatives that are in $(\lambda_1,\lambda_2)$ the {\it unresolved representatives} of $\cal M$.
For iteration $i = 1, 2, \ldots , \log n -1$,
let $S(i)$ be the number of submatrices in $\cal M$,
and $U(i)$ be the number of unresolved representatives of $\cal M$.
We first show that $S(i) \leq 7*2^{i+1}-4i-8$,
and $U(i) \leq 3*2^{i+3}-8i-14$.
We prove this by induction on $i$.
The basis is for $i=1$.
At the beginning of iteration 1,
there are 4 submatrices in $\cal M$ and 8 unresolved representatives of $\cal M$.
The first feasibility test resolves at least 4 of these representatives,
and the second feasibility test leaves at most 2 unresolved.
At most all 4 submatrices remain after discarding.
Splitting the submatrices at the end of iteration 1 gives at most 16 submatrices,
Note that for $i=1$, $S(i) \leq 16 = 7*2^{1+1}-4-8$ 
and 32 representatives, at most $32-6=26$ of which are unresolved.
and $U(i) \leq 26 = 3*16-8-14$.
Thus the basis is proved.

For the induction step, $i>1$.
By the induction hypothesis $S(i-1) \leq 7*2^i-4i-4$ and $U(i-1) \leq 3*2^{i+2}-8i-6$.
Let $R(i-1)$ be the set of representatives of $\cal M$ at the end of iteration $i-1$.
Note that these elements fall on at most $2^{i+1}-1$ diagonals of $M(P)$, since each iteration can only split the existing submatrices into at most four smaller submatrices. 
The feasibility tests on iteration $i$ will leave $u \leq \lfloor U(i-1)/4 \rfloor \leq 3*2^i-2i-2$ representatives unresolved.
Let $d_j$ be the number of elements of $R(i-1)$ on the $j$-th diagonal that are unresolved, so that $\sum d_j=u$. 
Except for possibly the submatrices with the largest and smallest representatives, all other submatrices on the $j$-th diagonal have two representatives, each in the range $(\lambda_1,\lambda_2)$. 
Thus, there will be at most $\lfloor (d_j +2)/2 \rfloor$ submatrices
whose representatives are on the $j$-th diagonal and are not
both at most $\lambda_1$ and not both at least $\lambda_2$.
Then summing over at most $2^{i+1}-1$ diagonals of $M(P)$, 
there are $S(i)/4 \leq \sum_j\lfloor{(d_j+2)/2\rfloor}\le(u + 2^{i+2}-2)/2$ submatrices
that cannot be discarded at the end of iteration $i$.
There will be $2S(i)/4 - u \leq 2^{i+2}-2$ representatives
of these submatrices that are resolved.
After quartering, there are $S(i) \leq 4*(u + 2^{i+2}-2)/2$
$\leq 3*2^{i+1}-4i-4 + 2^{i+3} -4$ submatrices
at the end of iteration $i$.
Simplifying, we have $S(i) \leq 7*2^{i+1}-4i-8$.
After quartering the submatrices,
the number of unresolved representatives of submatrices in $\cal M$ will be
$U(i) = 2S(i)-(2S(i)/4 - u)$
$= (3/2)S(i) + u$
$\leq (3/2)*(7*2^{i+1}-4i-8) + 3*2^i-2i-2$.
Simplifying,
we have $U(i) \leq 3*2^{i+3}-8i-14$.
This concludes the proof by induction.

There will be at most $\log n -1$ iterations until all submatrices consist
of single values.
At that point we will have
$S(\log n -1) \leq 7*2^{\log n}-4\log n -8$.
On each remaining iteration,
the number of elements in $\cal M$ will be at least quartered.
Thus the total number of iterations is $O(\log n)$.
The work on iteration $i$, for $i = 1, 2, \ldots ,\log n -1$,
will be $O(S(i))$ $= O(2^i)$.
Thus the total time on these iterations, exclusive of feasibility testing,
will be $O(\sum_{i=1}^{\log n -1}2^i)$ $=O(n)$.
The total time on the remaining iterations will be
$O(2(7*2^{\log n}-4\log n -8))$ $=O(n)$.
\end{proof}

\begin{figure}[thb]
\begin{center}
\includegraphics{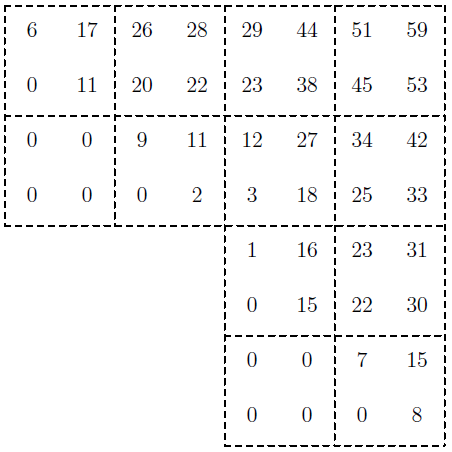}
\end{center}
\caption{\small Explicit illustration of submatrices, just after quartering on the first iteration of {\it PATH0}}
\label{fig2p2}
\end{figure}

\begin{figure}[thb]
\begin{center}
\includegraphics{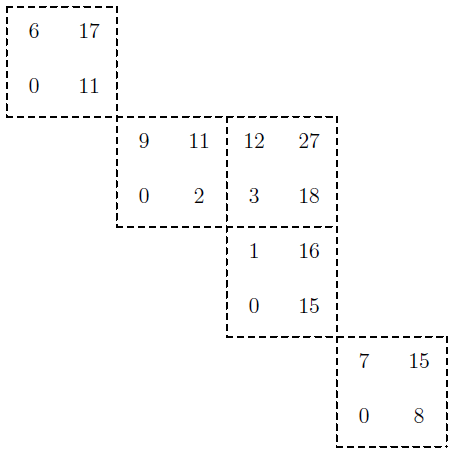}
\end{center}
\caption{\small Explicit illustration of submatrices at the end of the first iteration {\it PATH0}}
\label{fig2p3}
\end{figure}

We illustrate {\it PATH0} using path $P$ as in Fig. \ref{fig2p2}, with $k=3$.
First we set $\lambda_1$ to 0 and $\lambda_2$ to $\infty$.
We initialize the set $\cal M$ to the set consisting of
the four $4 \times 4$ submatrices of the matrix $M(P)$.
On the first iteration of the while-loop,
$R= \{59,3,28,0,31,0,0,0\}$ and $R'= \{59,3,28,31\}$.
The median of this set is $(28+31)/2 = 29.5$.
For $\lambda = 29.5$, no cuts are required, so we reset $\lambda_2$ to 29.5.
Then we recompute $R'$ to be $\{3,28\}$, whose median is $(3+28)/2 = 15.5$. 
For $\lambda = 15.5$, 1 cut is required, so we reset $\lambda_2$ to 15.5. 
We discard the submatrix with all values less than or equal to $\lambda_1$, leaving three submatrices. 
We quarter these submatrices into 12 submatrices as shown in Fig. \ref{fig2p2}.
Of these 12 submatrices, five have all values too large, and two have all values too small.
We discard them, leaving the five submatrices pictured in Fig. \ref{fig2p3}.

On iteration 2, $R= \{17,0,27,3,11,0,16,0,15,0\}$, and $R'= \{3,11,15\}$.
The median of this set is $11$.
For $\lambda = 11$, 3 cuts are required, so we reset $\lambda_1$ to 11.
Then we recompute $R'$ to be $\{15\}$, whose median is $15$.
For $\lambda = 15$, 2 cuts are required, so we reset $\lambda_2$ to 15.
There are no submatrices with all values at least 15.5, and one submatrix with all values at most 11, and we discard the latter.
We quarter the remaining four submatrices, giving 16 submatrices of dimension $1 \times 1$, of which all but the one containing $12$ are either too large or too small.
On iteration 3, $R=\{12\}$, $R'=\{12\}$, and the median is $12$.
For $\lambda = 12$, 3 cuts are required, so we reset $\lambda_1$ to 12.
At this point, all values are discarded, so that the revised $R'$ is empty, and a second selection is not performed on iteration 3.
All submatrices will be discarded from $\cal M$, and {\it PATH0} will terminate with $\lambda_1 = 12$ and $\lambda_2 = 15$, and output $\lambda^*=12$.
\bigskip

\begin{theorem}
\label{thm:2:2}
Algorithm {\it PATH0} finds a max-min partition of a path of $n$
weighted vertices in $O(n \log n)$ time.
\end{theorem}
\begin{proof}
Correctness follows from the correctness of {\it FTEST0}, from the fact that all possible candidates for $\lambda^*$ are included in $M(P)$, and from the fact that each value discarded is either at most $\lambda_1$ or at least $\lambda_2$.

By Lemma \ref{lem:2:1}, {\it PATH0} will take $O(n)$ total time, exclusive of the feasibility tests, and will produce a sequence of $O( \log n )$ values to be tested.
It follows that the total time for all feasibility tests is $O(n \log n)$.
\end{proof}

The time for {\it PATH0} corresponds to the times achieved by Megiddo and Cole
for the path problem.
We show how to do better in Section \ref{sec:path}.

We next describe a simple approach {\it TREE0} to the max-min problem on a tree.
We first define an {\it edge-path-partition} of a tree
rooted at a vertex of degree 1.
Partition the edges of the tree into paths,
where a vertex is an endpoint of a path
if and only if it is of degree not equal to 2 with respect to the tree.
Call any path in an edge-path-partition
that contains a leaf in the tree a {\it leaf-path}.
As an example, consider the vertex-weighted tree in Fig. \ref{tfig2p3}(a).
Fig. \ref{tfig2p3}(b) shows the edge-path-partition for $T$.
There are 7 paths in the partition, as shown.
Four of these paths are leaf-paths.

\begin{figure}[thb]
\begin{center}
\includegraphics[scale=0.9]{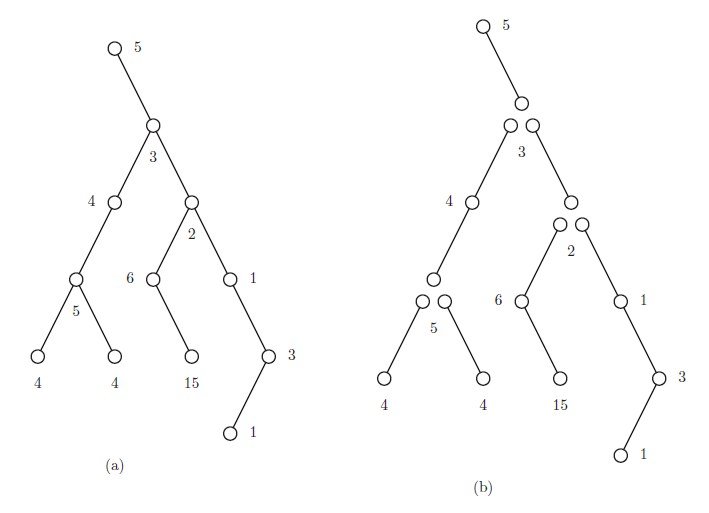}
\end{center}
\caption{\small An edge-path-partition for $T$}
\label{tfig2p3}
\end{figure}

\begin{figure}[thb]
\begin{center}
\includegraphics{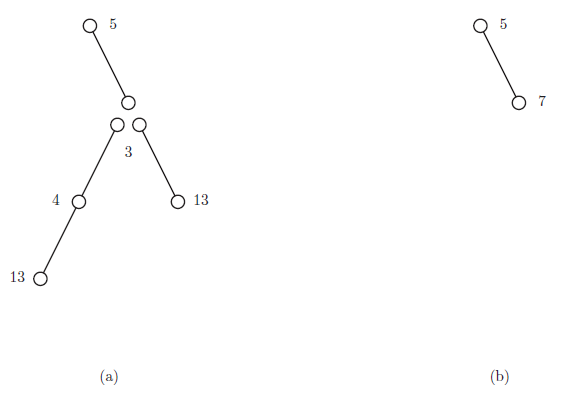}
\end{center}
\caption{\small The edge-path-partition of the resulting tree after all leaf-paths are deleted}
\label{tfig2p4}
\end{figure}

We now proceed with the approach for the tree.
Below are the three routines 
{\it TREE0\_init\_mat},
{\it TREE0\_test\_val},
and {\it TREE0\_update\_mat}.
The basic idea is to perform the search first on the leaf-paths,
and thus determine which edges in the leaf-paths should be cut.
When no search value on a leaf-path is contained in the open interval
$(\lambda_1, \lambda_2)$,
we prune the tree and repeat the process. 
We use the straightforward feasibility test described earlier.

The determination of cuts and pruning of the tree proceeds as follows.
If $T$ contains more than one leaf, do the following.
For each leaf-path $P_j$,
infer the cuts in $P_j$ such that each component in turn going up in $P_j$,
except the highest component on $P_j$, has total weight
as small as possible but greater than $\lambda_1$.
Delete all vertices beneath the top vertex of each leaf-path $P_j$,
and add to the weight of the top vertex in $P_j$
the weight of the other vertices in the highest component of $P_j$.
This leaves a smaller tree in which all leaf-paths
in the original tree have been deleted.
The smaller tree has at most half of the number of leaves of $T$.
Reset $T$ to be the smaller tree, and $k$ to be
the number of cuts remaining to be made in this tree.

\bigskip
\sspace
\noindent
{\it TREE0\_init\_mat}:\vspace{.05in}\\
$\T $ Initialize $T$ to be the tree rooted at a vertex of degree 1.\\
$\T $ Concatenate the leaf-paths of $T$ together, yielding path $P'$.\\
$\T $ Split sorted matrix $M(P')$ for the path $P'$ into 4 square submatrices. \\
$\T $ Let $\cal M$ be the set containing these four square submatrices.

\bigskip
\noindent
{\it TREE0\_test\_val}:\vspace{.05in}\\
$\T $ (Identical to {\it PATH0\_ident\_val}, \\
$\T \T$ except that {\it FTEST0} is called with argument $T$ rather than $P$).

\bigskip
\noindent
{\it TREE0\_update\_mat}:\vspace{.05in}\\
$\T$ Call {\it PATH0\_update\_mat} on $\mathcal{M}$.\\
$\T$ $\IF$ ${\cal M}$ is empty and $T$ is not a path \\
$\T$ $\TN$ \\
$\T \T$ $\FO$ each leaf-path of $T$ $\DO$ \\
$\T \T \T$ Determine cuts on the leaf-path, and decrease $k$ accordingly.\\
$\T \T \T$ Add to the weight of the top vertex the total weight of the\\
$\T \T \T \T$ other vertices in the highest component of the leaf-path.\\
$\T \T \T$ Delete from $T$ all vertices in the leaf-path except the top vertex.\\
$\T \T$ $\EF$ \\
$\T \T$ Concatenate the leaf-paths of $T$ together, yielding new path $P'$.\\
$\T \T$ Split sorted matrix $M(P')$ for path $P'$ into 4 square submatrices. \\
$\T \T$ Let $\cal M$ be the set containing these four square submatrices.\\
$\T$ $\EI$ 
 
\dspace
\bigskip

As an example we consider the max-min problem on the tree shown
in Fig. \ref{tfig2p3}(a), with $k=3$.
In the initialization,
four leaf-paths are identified and concatenated together,
giving the path $P' = 4, 5, 4, 5, 15, 6, 2, 1, 1, 3, 2$.
Once all search values associated with $P'$ are resolved,
$\lambda_1 = 10$ and $\lambda_2 = 13$. 
We place one cut between the vertices of weight 15 and 6, and reset $k$ to 2.
We add the weights of the two leaves of weight 4 to the weight of their parent, giving it weight 13.
We add the weights of all descendants of the vertex with weight 2, except the weight of 15, to the vertex of weight 2, giving it a weight also of 13.
Then we delete all edges in the leaf-paths. 
Fig. \ref{tfig2p4}(a) shows the edge-path-partition of the resulting tree.
There are two leaf-paths in this partition, with vertex weights 13, 4, and 3 on one, and 13 and 3 on the other.
We form the path $P' = 13, 4, 3, 13, 3$. 
Once we have resolved all search values associated with $P'$, we still have $\lambda_1 = 10$ and $\lambda_2 = 13$.
We can then place two cuts, above each of the vertices of weight 13,
and reset $k$ to 0. 
We add the weight of the vertex of weight 4 to its parent, giving it weight 7.
Then we delete all edges in the leaf-paths.
The edge-path-partition of the resulting tree is shown in Fig. \ref{tfig2p4}(b).
There is just one leaf-path in this partition,
with vertex weights 5 and 7.
Once all search values associated with $P'$ are resolved,
we have $\lambda_1 = 12$ and $\lambda_2 = 13$.
Since the tree consists of a single path, the algorithm then terminates with $\lambda^*=12$.
\bigskip

\begin{theorem}
\label{thm:2:3}
Algorithm {\it TREE0} finds a max-min partition of a tree of $n$
weighted vertices in $O(n (\log n)^2)$ time.
\end{theorem}
\begin{proof}
For correctness of the tree algorithm,
note that $\lambda_1$ always corresponds to a feasible value,
that all possible values resulting from leaf paths
are represented in $M(P')$,
and that the cuts are inferred on leaf-paths assuming (correctly)
that any subsequent value $\lambda$ to be tested
will have $\lambda > \lambda_1$.

We analyze the time as follows.
By Lemma \ref{lem:2:1}, resolving the path $P'$ will
use $O(\log n)$ feasibility tests,
and time exclusive of feasibility tests of $O(n)$.
Since each feasibility test takes $O(n)$ time,
the feasibility tests will use $O(n \log n)$ time.
Since resolving the path $P'$ will at least halve the number of leaves
in the tree, the number of such paths until the final version of $P'$ is $O(\log n)$.
Thus the total time to partition the tree by this method
is $O(n (\log n)^2)$.
\end{proof}

The time for {\it TREE0}
beats the time of $O(n(\log n)^3)$ for Megiddo's algorithm,
and matches the time of $O(n(\log n)^2)$ for Cole's algorithm.
We show how to do better for the max-min problem in Section \ref{sec:tree}.
\bigskip

\section{Partitioning a Path}
\label{sec:path}

In this section we present an optimal algorithm to perform
parametric search on a graph that is a path of $n$ vertices. 
The algorithm follows the paradigm of parametric search by repeatedly performing 
feasibility tests over potential optimal values, as it gathers 
information so that subsequent feasibility tests can be 
performed increasingly faster. The improvement in speed of the subsequent 
feasibility tests is enough so that the entire running time is linear.
Our discussion focuses on the max-min problem;
at the end of the section we identify the changes necessary
for the min-max problem.
Throughout this section we assume that the vertices of any path or subpath
are indexed in increasing order from the start to the end of the path.

We first consider the running time of {\it PATH0},
to determine how the approach might be accelerated.
All activities except for feasibility testing use a total of $O(n)$ time.
Feasibility testing can use a total of $\Theta (n \log n)$ time in the worst case,
since there will be $\Theta (\log n)$ values to be tested in the worst case,
and each feasibility test takes $\Theta (n)$ time.
It seems unlikely that one can reduce the number of tests that need to be made,
so then to design a linear-time algorithm,
it seems necessary to design a feasibility test that will quickly begin to take $o(n)$ time.
We show how to realize such an approach.

We shall represent the path by a partition into subpaths, each of which can be searched in time proportional to the logarithm of its length. 
Each such subpath will possess a property
that makes feasibility testing easier.
Either the subpath will be singular or resolved.
A subpath $P'$ is {\it singular} if it consists of one vertex,
and it is {\it resolved} if no value in $M(P')$ falls in the interval
$(\lambda_1, \lambda_2 )$.
If a subpath is resolved,
then the position of any one cut determines the positions
of all other cuts in the subpath,
irrespective of what value of $\lambda$ we choose from within
$(\lambda_1, \lambda_2 )$.
If we have arranged suitable data structures when the subpath becomes resolved,
then we do not need a linear scan of it.

To make the representation simple, we restrict the subpaths in the partition to have lengths that are powers of 2.
Each subpath will consist of vertices whose indices are
$(j-1)2^i+1, \ldots ,j\:2^i$ for integers $j>0$ and $i \geq 0$.
Initially the partition will consist of $n$ singular subpaths.
Each nonsingular subpath will have $i>0$.
When introduced into the partition,
each such subpath will replace its two {\it constituent subpaths},
the first with indices $(2j-2)2^{i-1}+1, \ldots ,(2j-1)2^{i-1}$,
and the second with indices $(2j-1)2^{i-1}+1, \ldots ,(2j)2^{i-1}$.

We represent the partition of path $P$ into the subpaths with 
three arrays $last[1..n]$, $ncut[1..n]$ and $next[1..n]$. 
Consider any subpath in the partition,
with first vertex $v_f$ and last vertex $v_t$.
Let $v_l$ be an arbitrary vertex in the subpath.
The array $last$ identifies the end of a subpath,
given the first vertex of a subpath.
Thus $last(l)=t$ if $l=f$ and is arbitrary otherwise.
Given a cut on the subpath,
the array $next$ identifies,
in constant time, a cut further on in that subpath.
The array $ncut$ identifies the number of cuts skipped
in moving to that further cut.
Let $w(l,t)$ be the sum of the weights of vertices $v_l$ through $v_t$.
If $w(l,t) < \lambda_2$,
then $next(l)=null$ and $ncut(l)=0$.
Otherwise, $next(l) \geq l$ is the index of the last vertex before a cut,
given that $l=1$ or $v_l$ is the first vertex after a cut.
Then $ncut(l)$ is the number of cuts after $v_l$ up to and including the one
following $v_{next(l)}$.
We assume that the last cut on a subpath will leave
a (possibly empty) subset of vertices of total weight less than $\lambda_2$.
(Note that the last cut on the path as a whole must then be ignored.)

Given the partition into subpaths, we describe feasibility test {\it FTEST1}.
Let $\lambda$ be the value to be tested,
with $\lambda_1 < \lambda < \lambda_2$.
For each subpath, we use binary search to find the first cut,
and then, follow $next$ pointers and add $ncut$ values to
identify the number of cuts on the subpath.
When we follow a path of $next$ pointers,
we will compress this path.
This turns out to be a key operation as we consider
subpath merging and its effect on feasibility testing.
(Note that the path compression makes
{\it FTEST1} a function with side effects.)\\

\sspace
\noindent
{\bf func} {\it FTEST1} ({\bf path} $P$, {\bf integer} $k$, {\bf real} $\lambda$){\vspace{.05in}\\
$\T $ $f \ASG 1$\\
$\T $ $numcut \ASG -1$; $remainder\ASG 0$\\
$\T $ $\WH$ $f \leq n$ $\DO$ /* search the next subpath: */\\
$\T \T $ $t \ASG last(f)$\\
$\T \T $ $\IF$ $remainder + w(f,t) < \lambda$\\
$\T \T $ $\TN$ $remainder \ASG remainder + w(f,t)$\\
$\T \T $ $\EL$ \\
$\T \T \T $ $numcut \ASG numcut+1$ \\
$\T \T \T $ Binary search for a smallest $r$ so that $w(f,r) + remainder \geq \lambda$. \\
$\T \T \T $ $\IF$ $r<t$ \\
$\T \T \T $ $\TN$ \\
$\T \T \T \T $ $(s,sumcut) \ASG$ {\it search\_next\_path} $(r,t)$\\
$\T \T \T \T $ $numcut \ASG numcut+sumcut$ \\
$\T \T \T \T $ {\it compress\_next\_path} $(r,s,t,sumcut)$ \\
$\T \T \T $ $\EI$  \\
$\T \T \T $ $remainder \ASG w(s+1,t)$ \\
$\T \T $ $\EI$  \\
$\T \T $ $f \ASG t+1$  \\
$\T $ $\EW$ \\
$\T $ $\IF$ $numcut \geq k$ $\TN$ {\bf return}(``$lower$'') $\EL$ {\bf return}(``$upper$'') $\EI$ \\
{\bf endfunc}
 
\bigskip
\noindent
{\it search\_next\_path} ({\bf vertex\_index} $l,t$)\vspace{.05in}\\
$\T $ $sumcut \ASG 0$ \\
$\T $ $\WH$ $l < t$ and $next(l+1) \neq null$ \\
$\T \T $ $sumcut \ASG sumcut+ncut(l+1)$ \\
$\T \T $ $l \ASG next(l+1)$ \\
$\T $ $\EW$ \\
$\T $ {\bf return}$(l,sumcut)$ 
 
\bigskip
\noindent
{\it compress\_next\_path} ({\bf vertex\_index} $l,s,t$, {\bf integer} $sumcut$)\vspace{.05in}\\
$\T $ $\WH$ $l < t$ and $next(l+1) \neq null$ \\
$\T \T $ $sumcut \ASG sumcut-ncut(l+1)$ \\
$\T \T $ $ncut(l+1) \ASG ncut(l+1)+sumcut$ \\
$\T \T $ $temp \ASG next(l+1)$ \\
$\T \T $ $next(l+1) \ASG s$ \\
$\T \T $ $l \ASG temp$ \\
$\T $ $\EW$  

\dspace
\bigskip

Use of this feasibility test by itself is not enough to guarantee
a quick reduction in the time for feasibility testing.
This is because there is no assurance that the interval
$(\lambda_1, \lambda_2)$ will be narrowed in a manner
that allows longer subpaths to quickly replace shorter subpaths
in the partition of path $P$.
To achieve this effect,
we reorganize the positive values from $M(P)$
into submatrices that correspond in a natural way to subpaths.
Furthermore, we associate synthetic weights with these submatrices
and use these weights in selecting the weighted median for testing.
The synthetic weights place a premium on resolving first the values from submatrices corresponding to short subpaths.
However, using synthetic weights could mean that we consider many values of small weight repeatedly, causing the total time for selection to exceed $\Theta (n)$.
To offset this effect, we also find the unweighted median, and test this value.
This approach guarantees that we discard at least half of the submatrices'
representatives on each iteration,
so that each submatrix inserted into $\cal M$ need be charged
only a constant for its share of the total work in
selecting values to test.

We now proceed to a broad description of {\it PATH1}.
The basic structure follows that of {\it PARAM\_SEARCH}, in that there will be the routines {\it PATH1\_init\_mat}, {\it PATH1\_test\_val}, and {\it PATH1\_update\_mat}. 
However, we will slip the set-up and manipulation
of the data structures for the subpaths into the routines
{\it PATH1\_init\_mat} and {\it PATH1\_update\_mat}.
Let $large(M)$ be the largest element in submatrix $M$,
and let $small(M)$ be the smallest element in $M$.\\
 
\sspace
\noindent
{\it PATH1\_init\_mat}:\vspace{.05in}\\
$\T $ Initialize $\cal M$ to be empty. \\
$\T $ Call {\it mats\_for\_path}$(P,1,n)$ to insert submatrices of $M(P)$ into $\cal M$. \\
$\T $ $\FO$ $l \ASG 1$ $\TO$ $n$ $\DO$ $last(l) \ASG l$; $next(l) \ASG 0$; $ncut(l) \ASG 0$ $\EF$
 
\bigskip
\noindent
{\it PATH1\_test\_val}:\vspace{.05in}\\
$\T $ $R \ASG \emptyset$ \\
$\T $ $\FO$ each $M$ in $\cal M$ $\DO$ \\
$\T \T $ $\IF$ $large(M) < \lambda_2$ \\
$\T \T $ $\TN$ Insert $large(M)$ into $R$ with synthetic weight $w(M)/4$. $\EI$ \\
$\T \T $ $\IF$ $small(M) > \lambda_1$ \\
$\T \T $ $\TN$ Insert $small(M)$ into $R$ with synthetic weight $w(M)/4$. $\EI$ \\
$\T $ $\EF$ \\
$\T$ Select the (synthetic) weighted median element $\lambda$ in $R$. \\
$\T$ $\IF$ {\it FTEST1}$(P,k,\lambda ) = ``lower$'' $\TN$ $\lambda_1 \ASG \lambda$ $\EL$ $\lambda_2 \ASG \lambda$ $\EI$ \\
$\T$ Remove from $R$ any values no longer in $(\lambda_1,\lambda_2)$. \\
$\T$ $\IF$ $R$ is not empty \\
$\T$ $\TN$ \\
$\T \T$ Select the unweighted median element $\lambda '$ in $R$. \\
$\T \T$ $\IF$ {\it FTEST1}$(P,k,\lambda ') = ``lower$'' $\TN$ $\lambda_1 \ASG \lambda '$ $\EL$ $\lambda_2 \ASG \lambda '$ $\EI$ \\
$\T$ $\EI$ 
 
\bigskip
\noindent
{\it PATH1\_update\_mat}:\vspace{.05in}\\
$\T $ $\WH$ there is an $M$ in $\cal M$ such that $small(M) \geq \lambda_2$ or $large(M) \leq \lambda_1$ \\
$\T \T $ or $small(M) \leq \lambda_1 \leq \lambda_2 \leq large(M)$ $\DO$ \\
$\T \T $ $\IF$ $small(M) \geq \lambda_2$ or $large(M) \leq \lambda_1$ \\
$\T \T $ $\TN$ \\
$\T \T \T $ Delete $M$ from $\cal M$. \\
$\T \T \T $ $\IF$ this is the last submatrix remaining for a subpath $P'$ \\
$\T \T \T $ $\TN$ $glue\_paths(P')$ \\
$\T \T \T $ $\EI$ \\
$\T \T $ $\EI$ \\
$\T \T $ $\IF$ $small(M) \leq \lambda_1$ and $large(M) \geq \lambda_2$ \\
$\T \T $ $\TN$ Split $M$ into four square submatrices, each of synthetic weight $w(M)/8$. \\
$\T \T $ $\EI$ \\
$\T $ $\EW$
 
\dspace
\bigskip

\begin{figure}[thb]
\begin{center}
\includegraphics{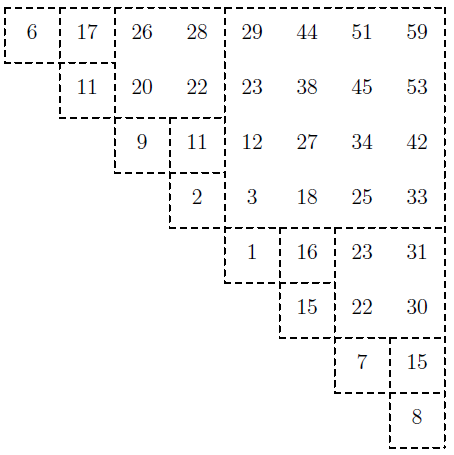}
\end{center}
{\caption{\small Explicit representation of initial submatrices for $M(P)$ in {\it PATH1}}\label{fig3p1}}
\end{figure}

Below is the procedure {\it mats\_for\_path},
which inserts submatrices of the appropriate previously discussed synthetic weight into $\cal M$.
A call with arguments $f$ and $t$
will generate submatrices for all required subpaths of
a path containing the vertices with indices $f,f+1,\ldots ,t$.
(In this section, we assume that $f=1$ and $t=n$,
but we state the procedure in this form so that it can be used
in the next section too.)
The submatrices for the matrix $M(P)$ in Fig.~\ref{fig1}
are shown in Fig.~\ref{fig3p1}. 
Every subpath $P'$ that will appear in some partition of $P$
is initialized with $cleaned(P')$ and $glued(P')$ to {\bf false},
where $cleaned(P')$ indicates whether or not all values in
the submatrix associated with $P'$ are outside the interval
$(\lambda_1,\lambda_2)$,
and $glued(P')$ indicates whether or not all values in $M(P')$
are outside the interval $(\lambda_1,\lambda_2)$. \\

\sspace
\noindent
{\bf proc} {\it mats\_for\_path} ({\bf path} $P$, {\bf integer} $f,t$){\vspace{.05in}\\
$\T $ $size \ASG 1$ \\
$\T $ $w \ASG 4n^4$ /* synthetic weight for $1\times1$ submatrices */\\
$\T $ $\WH$ $f\leq t$ $\DO$ \\
$\T \T $ $\FO$ $i \ASG f$ $\TO$ $t$ $\BY$ $size$ $\DO$ \\
$\T \T \T $ Insert the succinct description of submatrix $[i\:..\:(i\!+\!\lceil size/2\rceil\!-\!1),$ \\
$\T \T \T \T $ $(i\!+\!\lceil size/2\rceil )..\:(i\!+\!size\!-\!1)]$ for $M(P)$ into $\cal M$ with synthetic weight $w$. \\
$\T \T \T $ Let $P'$ be the subpath of $P$ whose vertices have indices \\
$\T \T \T \T $ $i,\ldots ,i+size-1$. \\
$\T \T \T $ $cleaned(P') \ASG \FA$; $glued(P') \ASG \FA$ \\
$\T \T \T $ $\EF$ \\
$\T \T $ $size \ASG size*2$ \\
$\T \T $ $w \ASG w/2$ /* smaller synthetic weight for increased size of submatrix */\\
$\T \T $ $f \ASG size*\lceil (f-1)/size\rceil +1$ \\
$\T \T $ $t \ASG size*\lfloor t/size\rfloor$ \\
$\T $ $\EW$  \\
{\bf endproc}\\

\dspace
\bigskip

We next describe the procedure {\it glue\_paths}, which checks to see if two constituent subpaths can be combined together.
Let $P'$ be any subpath of weight at least $\lambda_2$, and let $f'$ and $t'$ be the indices of the first and last vertices in $P'$, resp.
Let the $\lambda$-{\it prefix} of $P'$, designated $\lambda${\it pref}$\,(P')$, be vertices $v_{f'},\cdots ,v_l$ in $P'$ where $l$ is the largest index such that $w(f',l) < \lambda_2$.
Let the $\lambda$-{\it suffix} of $P'$, designated $\lambda${\it suff}$\,(P')$, be vertices $v_l,\cdots ,v_{t'}$ in $P'$ where $l$ is the smallest index such that $w(l,t') < \lambda_2$. 
To glue two subpaths $P_2$ and $P_3$ together into a subpath $P_1$, we must have $(glued(P_2)$ {\bf and} $glued(P_3)$ {\bf and} $cleaned(P_1))$. 
Procedure {\it glue\_paths} sets the $next$ pointers from vertices in $\lambda${\it suff}$\,(P_2)$ to vertices in $\lambda${\it pref}$\,(P_3)$. \\

\sspace
\noindent
{\bf proc} {\it glue\_paths} ({\bf path} $P_1$){\vspace{.05in}\\
$\T $ $cleaned(P_1) \ASG \TR$\\
$\T $ $\IF$ $P_1$ has length 1\\
$\T $ $\TN$ \\
$\T \T $ $glued(P_1) \ASG \TR$ \\
$\T \T $ Let $l$ be the index of the vertex in $P_1$. \\
$\T \T $ $\IF$ $w(l,l) \geq \lambda_2$ $\TN$ $next(l) \ASG l$; $ncut(l) \ASG 1$ $\EI$ \\
$\T \T $ Reset $P_1$ to be subpath of which $P_1$ is now a constituent subpath. \\
$\T $ $\EI$ \\
$\T $ Let $P_2$ and $P_3$ be the constituent subpaths of $P_1$. \\
$\T $ $\WH$ $glued(P_2)$ and $glued(P_3)$ and $cleaned(P_1)$ and $P_1 \neq P$ $\DO$ \\
$\T \T $ $glued(P_1) \ASG \TR$\\
$\T \T $ Let $f_2$ and $t_2$ be resp. the indices of the first and last vertices in $P_2$. \\
$\T \T $ Let $f_3$ and $t_3$ be resp. the indices of the first and last vertices in $P_3$. \\
$\T \T $ $last(f_2) \ASG t_3$\\
$\T \T $ $\IF$ $w(f_2,t_3) \geq \lambda_2$\\
$\T \T $ $\TN$ \\
$\T \T \T $ $\IF$ $w(f_2,t_2) < \lambda_2$ \\
$\T \T \T $ $\TN$ $f \ASG f_2$ \\
$\T \T \T $ $\EL$ /* initialize $f$ to the front of $\lambda${\it suff}$\,(P_2)$ */ \\
$\T \T \T \T $ $f \ASG t_2$\\
$\T \T \T \T $ $\WH$ $w(f-1,t_2) < \lambda_2$ $\DO$ $f \ASG f-1$ $\EW$ \\
$\T \T \T $ $\EI$ \\
$\T \T \T $ $t \ASG f_3$ \\
$\T \T \T $ $\WH$ $f \leq t_2$ and $w(f,t_3) \geq \lambda_2$ $\DO$ \\
$\T \T \T $ /* set $next$ pointers for $\lambda${\it suff}$\,(P_2)$ */ \\
$\T \T \T \T $ $\WH$ $w(f,t) < \lambda_2$ $\DO$ $t \ASG t+1$ $\EW$ \\
$\T \T \T \T $ $ncut(f) \ASG 1$ \\
$\T \T \T \T $ $next(f) \ASG t$ \\
$\T \T \T \T $ $f \ASG f+1$ \\
$\T \T \T $ $\EW$ \\
$\T \T $ $\EI$  \\
$\T \T $ Reset $P_1$ to be subpath of which $P_1$ is now a constituent subpath. \\
$\T \T $ Let $P_2$ and $P_3$ be the constituent subpaths of $P_1$. \\
$\T $ $\EW$ \\
{\bf endproc}\\

\dspace

It would be nice if procedure {\it glue\_paths},
after setting pointers in $\lambda${\it suff}$\,(P_2)$,
would perform pointer jumping,
so that $next$ pointers for vertices in $\lambda${\it pref}$\,(P_1)$
would point to vertices in $\lambda${\it suff}$\,(P_1)$. 
Unfortunately, it is not apparent how to incorporate pointer jumping into {\it glue\_paths} without using $O(n)$ time over all calls in the worst case. 
We thus opt for having {\it glue\_paths} do no pointer jumping,
and instead we do pointer jumping under the guise of path
compression in {\it FTEST1}.
We use an argument based on amortization to show that this works well.
\bigskip

\begin{lemma}
\label{lem:3:1}
Let $P$ be a path of $n$ vertices.
All calls to {\it glue-paths} will take amortized time of $O(n)$,
and {\it FTEST1} will search each subpath that is glued
in amortized time proportional to the logarithm of its length.
\end{lemma}
\begin{proof}
Suppose two subpaths $P_2$ and $P_3$ are glued together to give $P_1$,
where $P_1$ has weight at least $\lambda_2$.
We consider two cases.
First suppose either $P_2$ or $P_3$ has weight at most $\lambda_1$.
Let $P_m$ represent this subpath.
The time for {\it glue\_paths} is in worst case proportional
to the length of $P_m$.
Also, we leave a {\it glue-credit}
on each vertex in the $\lambda${\it pref}$\,(P_1)$ and $\lambda${\it suff}$\,(P_1)$,
and a {\it jump-credit} on each vertex in the $\lambda${\it pref}$\,(P_1)$.
The number of credits will be proportional to the length of $P_m$.
We charge this work and the credits to the vertices of $P_m$,
at a constant charge per vertex.
It is clear that any vertex in $P$ is charged at most once,
so that the total charge to vertices over all calls to {\it glue\_paths}
is $O(n)$.
The second case is when both $P_2$ and $P_3$ have weight at least $\lambda_2$.
In this case, the time for {\it glue\_paths} is proportional to the sum
of the lengths of $\lambda${\it suff}$\,(P_2)$
and $\lambda${\it pref}$\,(P_3)$,
and we use the glue-credits of $P_2$ and $P_3$ to pay for this.
The jump-credits are used by {\it FTEST1} rather than by {\it glue-paths},
and in fact, would present a problem if one tries to use them in {\it glue-paths}, as we discuss below.

Suppose both $P_2$ and $P_3$ have weight at least $\lambda_2$.
Then we might have wanted to have {\it glue-paths} jump the $next$ pointers
so that the $next$ pointer for a vertex in $\lambda${\it pref}$\,(P_2)$
would be reset to point to $\lambda${\it suff}$\,(P_3)$.
The time to reset such pointers would be proportional to the length
of $\lambda${\it pref}$\,(P_2)$.
The jump-credits of $\lambda${\it pref}$\,(P_3)$ could pay for this,
as long as the length of $\lambda${\it pref}$\,(P_2)$
is at most some constant (say 2)
times the length of the $\lambda${\it pref}$\,(P_3)$.
When the length of the $\lambda${\it pref}$\,(P_2)$
is at most twice the length of the $\lambda${\it pref}$\,(P_3)$,
we would not want to jump the pointers.

In general we could view a subpath as containing a sequence of
$\lambda$-{\it regions},
where each $\lambda$-region was once the $\lambda$-prefix of some subpath,
and the length of a $\lambda$-region is less than twice the length
of the preceding $\lambda$-region.
Note that the number of $\lambda$-regions in the sequence could be at most
the logarithm of the length of the subpath.
When gluing two subpaths together,
we could concatenate their sequences of $\lambda$-regions together,
and jump pointers over any $\lambda$-region that gets a predecessor
whose length is not more than twice its length.
Its jump-credits could then be used for this pointer jumping.
Searching a subpath in {\it FTEST1} would then take
time at most proportional to its length.

Of course {\it glue-paths} does not jump pointers.
However,
a lazy form of pointer-jumping is found in the path-compression of {\it FTEST1},
and the same sort of analysis can be seen to apply.
Suppose that {\it FTEST1} follows a pointer to a vertex
that was once in the $\lambda$-prefix of some subpath.
Consider the situation if {\it glue-paths} had jumped pointers.
If that pointer would have been present in the sequence of $\lambda$-regions,
then charge the operation of following that pointer to {\it FTEST1}.
Otherwise, charge the operation of following that pointer
to the jump-credits of the $\lambda$-prefix containing the vertex.
It follows that the number of pointers followed during a search
of a subpath that are not covered by jump-credits is at most
the logarithm of the length of the subpath.
Also, the binary search to find the position of the first cut
on a subpath uses time at most proportional
to the logarithm of the length of the subpath.
\end{proof}

\begin{lemma}
\label{lem:3:2}
Let $P$ be a path of $n$ vertices.
On the $i$-th iteration of the while-loop of {\it PATH1},
the amortized time used by feasibility test {\it FTEST1}
will be $O(i(5/6)^{i/5}\:n)$.
\end{lemma}

\begin{proof}
We first consider the synthetic weights assigned to submatrices in $\cal M$.
Corresponding to subpaths of lengths $1, 2, 4, \ldots , n$,
procedure {\it mats\_for\_path} creates
$n$ submatrices of size $1 \times 1$,
$n/2$ submatrices of size $1 \times 1$,
$n/4$ submatrices of size $2 \times 2$, and so on,
up through one submatrix of size $n/2 \times n/2$.
The total synthetic weight of the submatrices corresponding to each path length is
$4n^5, n^5, n^5/4, \ldots , 4n^3$, resp.
It follows that the total synthetic weight for submatrices of all path lengths is
less than $(4/3)*4n^5$.

Let $wgt(M)$ be the synthetic weight assigned to submatrix $M$.
Define the {\it effective weight},
denoted {\it eff\_wgt}$(M)$, of a submatrix $M$ in $\cal M$
to be $wgt(M)$ if both its smallest value and largest value
are contained in the interval $(\lambda_1, \lambda_2)$
and $(3/4)wgt(M)$ if only one of its smallest value and largest value
is contained in the interval $(\lambda_1, \lambda_2)$.
Let {\it eff\_wgt}$(\cal M )$ be the total effective weight of all submatrices in $\cal M$.
When a feasibility test renders a value (the smallest or largest)
from $M$ no longer in the interval $(\lambda_1, \lambda_2)$,
we argue that {\it eff\_wgt}$(M)$ is reduced by at least $wgt(M)/4$
because of that value.

If both values were contained in the interval $(\lambda_1, \lambda_2)$,
and one no longer is, then clearly our claim is true.
If only one value was contained in the interval $(\lambda_1, \lambda_2)$,
and it no longer is,
then $M$ is replaced by four submatrices of effective weights
$wgt(M)/8 + wgt(M)/8 + (3/4)wgt(M)/8 + (3/4)wgt(M)/8$ $< wgt(M)/2$,
so that there is a reduction in effective weight by greater than $wgt(M)/4$.
If both values were contained in the interval $(\lambda_1, \lambda_2)$,
and both are no longer in,
then we consider first one and then the other.
Thus every element that was in $(\lambda_1, \lambda_2)$ but no longer is
causes a decrease in {\it eff\_wgt}$(\cal M )$ by an amount
at least equal to its synthetic weight in $R$.
Values with at least half of the total weight in $R$
find themselves no longer in $(\lambda_1, \lambda_2)$.
Furthermore, the total weight of $R$ is at least $1/3$ of
{\it eff\_wgt}$(\cal M )$.
This follows since a submatrix $M$ has either two values in $R$
at a total of $2w(M)/4 = (1/2)w(M)$
or one value at a weight of $w(M)/4 = (1/3)(3w(M)/4)$.
Thus {\it eff\_wgt}$(\cal M )$
decreases by a factor of at least $(1/2)(1/3) = 1/6$ per iteration.

Corresponding to the subpaths of length $2^j$,
there are $n/2^j$ submatrices created by {\it mats\_for\_path}.
If such a matrix is quartered repeatedly until $1 \times 1$ submatrices result,
each such submatrix will have weight $n^4/2^{4j-2}$.
Thus when as little as $(n/2^j)*(n^4/2^{4j-2})$ synthetic weight remains,
all subpaths of length $2^j$ can still be unresolved.
This can be as late as iteration $i$,
where $i$ satisfies
$(4/3)*4n^5*(5/6)^i = n^5/2^{5j-2}$,
or $2^{5j}=(3/4)*(6/5)^i$.
While all subpaths of length $2^j$ can still be unresolved on this iteration,
at most $(1/2*1/2^4)^k = 1/2^{5k}$ of the subpaths of length $2^{j-k}$
can be unresolved for $k= 1, \ldots, j$.
Also, by Lemma~3.1, each subpath can be searched by {\it FTEST1}
in amortized time proportional to the logarithm of its length.
Thus the time to search path $P$ on iteration $i$ is at worst proportional to
$(j/2^j)n(1 + 1/2^5 + 1/2^{10} + \cdots )$,
which is $O((j/2^j)n)$.
From the relationship of $i$ and $j$,
$2^j = (3/4)^{1/5}(6/5)^{i/5}$,
and $j = (1/5)\log(3/4) + (i/5)\log (6/5)$.
The lemma then follows.
\end{proof}

\begin{lemma}
\label{lem:3:3}
The total time for handling $\cal M$ and
performing selection over all iterations of {\it PATH1} is $O(n)$.
\end{lemma}
\begin{proof}
Using the algorithm of \cite{BFPRT},
the time to perform selection in a given iteration
is proportional to the size of $R$.
Define the {\it effective count}, denoted {\it eff\_cnt}$(M)$,
of a submatrix $M$ in $\cal M$ to be
2 if both its smallest value and largest value
are contained in the interval $(\lambda_1, \lambda_2)$
and 1 if only one of its smallest value and largest value
is contained in the interval $(\lambda_1, \lambda_2)$.
Let {\it eff\_cnt}$(\cal M )$ be the total effective count
of all submatrices in $\cal M$.
Then the size of $R$ equals {\it eff\_cnt}$(\cal M )$.
On any iteration, the result of the feasibility test of $\lambda '$
is to resolve at least half of the values in $R$.

The time for inserting values into $R$ and performing the selections is $O(n)$.
We show this by using an accounting argument, charging 2 credits for each value
inserted into $R$. As $R$ changes, we maintain the invariant that the number
of credits is twice the size of $R$. When $R$ has $k$ elements,
a selection takes $O(k)$ time, paid for by $k$ credits, leaving $k/2$ elements
covered by $k$ credits. Since $n$ elements are inserted into $R$ during the whole of
{\it PATH1}, the time for forming $R$ and performing selections is $O(n)$.

It remains to count the number of submatrices inserted into $\cal M$.
Initially $2n-1$ submatrices are inserted into M. 
For $j = 1, 2, \ldots ,\log n -1$,
consider all submatrices of size $2^j \times 2^j$ that are at some point in $\cal M$. 
A matrix that is split must have its smallest value at most $\lambda_1$ and its largest value at
least $\lambda_2$. However, $M_{i,j} > M_{i-k,j+k}$ for $k > 0$, since the path represented
by $M_{i+k,j-k}$ is a subpath of the path represented by $M_{i,j}$. Hence, for any $2^j\times 2^j$
submatrix that is split, at most one submatrix can be split
in each diagonal going from lower left to upper right. There are fewer than
$2n$ diagonals, so there will be fewer than $2(n/2^j)$ submatrices that are split.
Thus the number resulting from quartering is less than $8(n/2^j)$. Summing
over all $j$ gives $O(n)$ submatrices in $\cal M$ resulting from quartering.
\end{proof}

We illustrate {\it PATH1} on path $P$ in Fig.~\ref{fig1}, with $k=3$.
(This is the same example that we discussed near the end of the previous section.)
The initial submatrices for ${\cal M}$ are shown in Fig.~\ref{fig3p1}.
The submatrix of size $4 \times 4$ has synthetic weight $2^{11}$,
the two submatrices of size $2 \times 2$ have synthetic weight $2^{12}$,
four submatrices of size $1 \times 1$ have synthetic weight $2^{13}$,
and the remaining eight submatrices of size $1 \times 1$ have synthetic weight $2^{14}$.
Initially $\lambda_1 = 0$ and $\lambda_2 = \infty$.
The $last$, $next$ and $ncut$ arrays are initialized as stated.

On the first iteration,
$R$ contains two copies each of
$6,11,9,2,1,15,7,8$ of synthetic weight $2^{12}$,
two copies each of $17,11,16,15$ of synthetic weight $2^{11}$,
$20,28,22,31$ of synthetic weight $2^{10}$,
and $3,59$ of synthetic weight $2^{9}$.
The weighted median of $R$ is 11. 
For $\lambda = 11$, three cuts are required, so we reset $\lambda_1$ to 11.
The revised version of $R$ will then be
$\{15,15,17,17,16,16,15,15,20,28,22,$ $31,59\}$.
The median of this is 17.
For $\lambda '= 17$, one cut is required, so we reset $\lambda_2$ to 17.
The set $\cal M$ is changed as follows. 
We discard all submatrices of size $1 \times 1$ except the one of synthetic weight $2^{14}$ containing $15$, and the two of synthetic weight $2^{13}$ containing $15$ in one and $16$ in the other.
We discard all submatrices of size $2 \times 2$.
We quarter the submatrix of size $4 \times 4$, giving four submatrices each of synthetic weight $2^8$.
We discard three of these submatrices because their values are all too large.
We quarter the remaining submatrix (containing $27,12,18,3$), giving four submatrices each of synthetic weight $2^5$.
We discard all submatrices but the submatrix of size $1 \times 1$ containing $12$.

As a result of submatrix discarding, we render a number of subpaths cleaned and glued.
We clean and glue every subpath of length 1 except $v_6$, the subpaths $v_1,v_2$ and $v_3,v_4$, the subpath $v_1,v_2,v_3,v_4$, and we clean but do not glue the subpath $v_5,v_6,v_7,v_8$.
When we glue subpath $v_6$, we set $next(6)$ to $6$, and $ncut(6)$ to $1$.
When we glue subpath $v_1,v_2$, we set $last(1)$ to $2$, $next(1)$ to $2$, and $ncut(1)$ to $1$.
When we glue subpath $v_3,v_4$, we set $last(3)$ to $4$, but we change no $next$ or $ncut$ values.
When we glue subpath $v_1,v_2,v_3,v_4$, we set $last(1)$ to $4$, $next(2)$ to $3$, and $ncut(2)$ to $1$. 
This completes all activity on the first iteration.

On the second iteration,
$R$ contains two copies of
$15$ of synthetic weight $2^{12}$,
two copies each of $16,15$ of synthetic weight $2^{11}$,
and two copies of $12$ of synthetic weight $2^{3}$.
The weighted median of $R$ is 15.
For $\lambda = 15$, 2 cuts are required,
so we reset $\lambda_2$ to 15.
The revised version of $R$ will then be
$\{12,12\}$.
The median of this is 12.
For $\lambda '= 12$, 3 cuts are required,
so we reset $\lambda_1$ to 12.

All remaining submatrices are discarded from $\cal M$.
When this happens, we clean and glue all remaining subpaths.
When we glue subpath $v_5,v_6$, we set $last(5)$ to $6$, $next(5)$ to $6$, and $ncut(5)$ to $1$.
When we glue subpath $v_7,v_8$, we set $last(7)$ to $8$, $next(7)$ to $8$, and $ncut(7)$ to $1$.
When we glue subpath $v_5,v_6,v_7,v_8$, we set $last(5)$ to $8$,
but we change no $next$ or $ncut$ values.
When we glue subpath $v_1,v_2,v_3,v_4,v_5,v_6,v_7,v_8$, we set $last(1)$ to $8$, $next(3)$ and $next(4)$ to $6$,
and $ncut(3)$ and $ncut(4)$ to $1$.
Since $\cal M$ is empty,
{\it PATH1} will then terminate with $\lambda_1 = 12$, 
$\lambda_2 = 15$, and output $\lambda^*=12$.

Note that we performed no compression of search paths 
on the second iteration.
Suppose for the sake of example that we performed a subsequent search
with $\lambda = 13$.
The initial cut would come after $v_2$,
and $next(3)$ and $next(7)$ would be followed to arrive at $v_8$.
Then we would reset $ncut(3)$ to $2$, and $next(3)$ to $8$.

\begin{theorem}
\label{thm:3:4}
Algorithm {\it PATH1} solves the max-min $k$-partitioning problem
on a path of $n$ vertices in $O(n)$ time.
\end{theorem}
\begin{proof}
The use of arrays $next$ and $ncut$, indexed by vertices on the
path, is the key difference between {\it PATH0} and {\it PATH1}. Feasibility test
{\it FTEST1} requires $next(a)$ to point to a further vertex, $b$, in the subpath,
such that $ncut(a)$ is the number of cuts between $a$ and $b$. Enforcing this 
requirement, we update arrays $next$ and $ncut$ by two subroutines of {\it PATH1}.

The first is {\it glue\_paths}. 
When $small(M) \geq \lambda_2$ or $large(M) \leq \lambda_1$, certain vertices on the subpaths need no longer be inspected, so {\it glue\_paths} combines two adjacent subpaths of equal length by updating the $next$ and $ncut$ values on the subpaths. 
We repeat the gluing and updating until no longer possible.
The second subroutine is {\it compress\_next\_path}, which we call after
{\it search\_next\_path}$(a, b)$ in {\it FTEST1}. Procedure {\it search\_next\_path}$(a, b)$ 
determines $(v, sumcut)$, where $v$ is the location of the final cut before $b$ and
$sumcut$ is the number of cuts between $a$ and $v$. 
We then update the values of $next$ and $ncut$ for vertices between $a$ and $v$.

The correctness of {\it FTEST1} follows, as we increment $numcut$ once for each subpath whose weight exceeds $\lambda^*$, or by $sumcut$ for each compressed subpath with the corresponding number of cuts. 
Correctness of {\it PATH1} follows from the correctness of {\it FTEST1}, from the
fact that all possible candidates for $\lambda^*$ are included in $\cal M$, and from the fact
that each value discarded is either at most $\lambda_1$ or at least $\lambda_2$.

The time to initialize $\cal M$ is clearly $O(n)$.
By Lemma \ref{lem:3:3},
the total time to select all values to test for feasibility is $O(n)$.
By Lemma \ref{lem:3:2},
the amortized time to perform feasibility test on iteration $i$
is $O(i(5/6)^{i/5}n)$.
Summed over all iterations,
this quantity is $O(n)$.
By Lemma \ref{lem:3:3},
the total time to handle submatrices in $\cal M$ is $O(n)$.
By Lemma \ref{lem:3:1},
the time to manipulate data structures for the subpaths is $O(n)$.
The time bound then follows.
\end{proof}

We briefly survey the differences needed to solve the min-max path-partitioning problem.
Procedure {\it FTEST1} is similar except that
$r$ is the largest index such that
$w(f,r) + remainder \leq \lambda$,
we subtract $1$ from $numcut$ after completion of the while-loop if $remainder = 0$,
and use $numcut \geq k$ rather than $numcut > k$. 
Upon termination, {\it FTEST1} outputs $\lambda^*=\lambda_2$. 
Procedure {\it glue\_paths} is similar, except that we replace the statement
\vspace{.1in}\\
$\T $ $\WH$ $w(f,t) < \lambda_2$ $\DO$ $t \ASG t+1$ $\EW$
\vspace{.1in}\\
by the statement
\vspace{.1in}\\
$\T $$\WH$ $w(f,t+1) < \lambda_2$ $\DO$ $t \ASG t+1$ $\EW$
\vspace{.1in}\\
Note that $w(l,l) \leq \lambda_2$ always holds for the min-max problem.
\bigskip

\begin{theorem}
\label{thm:3:5}
The min-max $k$-partitioning problem can be solved
on a path of $n$ vertices in $O(n)$ time.
\end{theorem}
\begin{proof}
The above changes will not affect the asymptotic running time of
algorithm {\it PATH1}.
\end{proof}

\section{Partitioning for Max-Min on a Tree}
\label{sec:tree}
In this section, we present an optimal algorithm to perform parametric search for the max-min problem on a graph that is a tree. 
The algorithm differs from that in Section~\ref{sec:path} as either long paths or many leaves can overwhelm the running time, so we must simultaneously compress long paths and delete leaves.  
The situation is further complicated by the challenge of handling \emph{problematic vertices}, which are vertices of degree greater than two. 
Thus, we pursue a dual-pronged strategy that identifies both paths to compress and paths to delete. 

\subsection{Our Tree Algorithm}
Our algorithm proceeds through rounds, running an increasing number of selection and feasibility tests on each round. 
As we shall show in Section \ref{sec:tree:analysis}, each round halves the feasibility test time, and the overall time for selection is linear. 
Thus, our algorithm runs in linear time.

For the purposes of discussion and analysis, we initially classify paths in the tree as either pending paths or processed paths, based on whether the actual weights associated with the path have already been resolved. 
Let an \emph{internal path} be a subpath between the root and a problematic vertex, or between two problematic vertices such that all intermediate vertices are of degree $2$. 
Let a \emph{leaf path} be a subpath with either the root or a problematic vertex as one endpoint and the other endpoint being a leaf. 
Furthermore, let a \emph{processed path} be an internal path that is completely cleaned and glued and let a \emph{pending} path be a path, either an interval or a leaf path, that contains some value in the interval $(\lambda_1,\lambda_2)$. 
Moreover, we define and classify paths or subpaths as light paths, middleweight paths, or heavy paths, based on how the actual weights in the path compare to $\lambda_1$ and $\lambda_2$. 
A subpath is \emph{light} if its total weight is at most $\lambda_1$, and a subpath is \emph{heavy} if its total weight is at least $\lambda_2$. 
Otherwise, a subpath is \emph{middleweight}. 
Note that a middleweight path may consist of a sequence of both light and/or middleweight subpaths.

For each round, our algorithm runs an increasing number of selection and feasibility tests, while maintaining three selection sets, $\mathcal{H}$, $\mathcal{U}$, and $\mathcal{V}$, consisting of candidates for the max-min value. 
The first selection set is for heavy subpaths, the second is for middleweight paths formed as sequences of light or middleweight subpaths after the resolution of a problematic vertex, and the third is for handling problematic vertices whose leaf paths have been resolved. 
We assign synthetic weights to values inserted into $\mathcal{H}$ or $\mathcal{U}$ as a function of the lengths of the corresponding paths, as defined in the procedure {\it mats\_for\_path} in Section~\ref{sec:path}, assuming that we round up the length to a power of $2$.  

Our algorithm invokes a separate procedure to address each of these sets specifically. 
Procedure {\it handle\_middleweight\_paths} inserts the total actual weight of each middleweight path as elements into $\mathcal{U}$. 
Our algorithm then performs two weighted selections on the elements in $\mathcal{U}$, one using the length of each path as the weight, and the other using the synthetic weight of each path. 
For each of the weighted selections, our algorithm then tests the selected value for feasibility, and adjusts $\lambda_1$ and $\lambda_2$ accordingly. 
For any middleweight path that becomes heavy, our algorithm succinctly identifies the corresponding submatrix and inserts the representatives (of the submatrix) whose values are within $(\lambda_1,\lambda_2)$ into $\mathcal{H}$.
For any path that becomes light, our algorithm cleans and glues the path, so that future feasibility tests require only polylogarithmic time to search the path.

The second procedure, {\it handle\_pending\_paths}, selects the weighted and unweighted medians from $\mathcal{H}$ separately, tests each for feasibility, and adjusts $\lambda_1$ and $\lambda_2$ accordingly. 
After the feasibility test, our algorithm cleans and glues adjacent subpaths that it has just resolved. 
When a leaf path is completely resolved, our algorithm represents the leaf path by a single vertex with the remaining accumulated weight, $accum\_wgt(v)$, as described in Section~\ref{sec:prelims}. 
Our algorithm defers until the next iteration of the for-loop the insertion into $\mathcal{H}$ of the representatives of any subpaths that have become heavy.

The third procedure is {\it handle\_leaves}. 
For any problematic vertex with a resolved leaf path hanging off it, the procedure inserts into $\mathcal{V}$ the sum of the weight of that vertex plus the accumulated remaining weight left over from the resolved leaf path. 
The procedure then selects the median from $\mathcal{V}$, performs the feasibility test, and adjusts $\lambda_1$ and $\lambda_2$ accordingly. 
Since we want to find the max-min, if the number of cuts is at least $k$, then for each vertex whose weight plus the accumulated remaining weight from the resolved leaf path is at most the median, our algorithm merges the vertex with the accumulated remaining weight from the resolved leaf path. 
On the other hand, if the number of cuts is less than $k$, then for each vertex whose weight plus the accumulated remaining weight from the resolved leaf path is at least the median, our algorithm cuts below the parent vertex, because any future feasibility tests would do the same.  
Furthermore, our algorithm assigns, to any middleweight path created by the resolution of a problematic vertex, a synthetic weight that is a function of the length (rounded up to a power of two) as defined in the procedure {\it mats\_for\_path} in Section~\ref{sec:path}. 
Our algorithm then inserts the weight of that path into $\mathcal{U}$ at the beginning of the next invocation, along with the corresponding synthetic weight. 
Note that we merge middleweight paths only when we know whether they will become light or heavy. 
For the representative of any subpath that becomes heavy during this procedure, we wait until the next iteration of the for-loop to insert the representative into $\mathcal{H}$.

We now give the top level of our algorithm. 
We defer our discussion of the corresponding data structures until Section~\ref{sec:tree:structures}.
\vskip 0.2in\noindent
\sspace
{\it TREE1}:\vspace{.05in} \\
$\T $ Initialize the data structures for the algorithm and set round $r\leftarrow1$. \\
$\T $ $\WH$ $\mathcal{H}\cup\mathcal{U}\cup\mathcal{V}\neq\emptyset$ $\DO$ \\
$\T \T $ $\WH$ feasibility test time is more than $n/2^r$ $\DO$ \\
$\T \T \T$ Call procedure {\it handle\_middleweight\_paths}.\\
$\T \T \T$ $\REPEAT$ Call procedure {\it handle\_pending\_paths} \\
$\T \T \T$ $\UNTIL$ the feasibility test time on heavy paths has been reduced by $50\%$ \\
$\T \T \T \T$ $\AND$ the feasibility test time on leaf paths has been reduced by $25\%$ \\
$\T \T \T$ Call procedure {\it handle\_leaves}.\\
$\T \T $ $\EW$ \\
$\T \T$ $r\leftarrow r+1$\\
$\T $ $\EW$
\bigskip
\subsection{Data Structures}
\dspace
\label{sec:tree:structures}
As our algorithm progresses, it prunes the tree to delete some of the paths in the edge-path-partition and it glues together other paths, while updating the values of $\lambda_1$ and $\lambda_2$.
In this section, we discuss the necessary data structures so that our algorithm can efficiently perform feasibility testing as these updates occur.

Our algorithm represents a path in an edge-path-partition in consecutive locations of an array. 
To achieve that, our algorithm initially stores the actual weights of the vertices of the tree in an array using the following order. 
It orders the children of each vertex from left to right by nondecreasing height in the tree. 
Using this organization of the tree, our algorithm then lists the vertices of the tree in postorder. 
Our algorithm can find the heights of all vertices in linear time. 
It can also order pairs containing parent and height lexicographically in linear time. 
Our algorithm does this process only during the initialization phase. 
It uses arrays that are parallel to the actual weight array to store pointer values, as well as the accumulated weight (see procedure \emph{explore} in Section~\ref{sec:prelims}), which it then uses to compute the actual weight of specified subpaths, exactly the same as the $last$, $ncut$, and $next$ arrays in Section~\ref{sec:path}.
We discuss this additional information in due course.

When our algorithm removes paths from the tree, it reorganizes the storage of the tree within the array as follows. 
Let $P$ be a leaf-path in the edge-path-partition of the current tree, and let $t$ be the top vertex of $P$. 
To remove $P-t$, we proceed as follows: 
Let $P'$ be the path whose bottom vertex is $t$. 
If $t$ will have only one child remaining after removal of $P-t$, and this child was originally not the rightmost child of $t$, do the following. 
Let $P''$ be the other path whose top vertex is $t$. 
Copy the vertices of $P''-t$ in order, so that the top vertex of this path is in the location preceding the location of $t$ in the array. 
Modify the actual weight of $t$ by adding the remainder left from removing $P-t$. 
Also, copy pointer values for all copied vertices into the new locations in their arrays. 
That is, update $last$, $ncut$, and $next$ pointers in $P'$. 
In particular, the $last$ pointer for the first vertex in $P'$ should now point to the last vertex in $P''$. 
We should also copy and modify the accumulated weight, as we discuss shortly.

Note that the bottom vertex of $P''$ may have been the parent of several vertices. 
When $P''-t$ is moved, we do not copy the children of its bottom vertex (and subtrees rooted at them). 
It is simple to store in a location formerly assigned to the bottom vertex of $P''$ the current location of this vertex, and to also store a pointer back. 
If we copy the path containing $P''$, then we reset this pointer easily. 
When only one child of the bottom vertex of $P''$ remains, we copy the corresponding path to in front of $P''$. 
We claim that the total time to perform all rearrangements will be $O(n)$: 
The time to copy each vertex and to copy and adjust its accumulated weight is constant. 
Because of the way in which the tree is stored in the array, at most one vertex will be copied from any array location.

We next discuss the representation of a heavy path. 
Each heavy path $P$ is represented as a sequence of overlapping subpaths, each of actual weight at least $\lambda_2$. 
Each vertex of $P$ is in at most two overlapping subpaths. 
Each overlapping subpath, except the first, overlaps the previous overlapping subpath on vertices of total actual weight at least $\lambda_2$, and each overlapping subpath, except the last, overlaps the following overlapping subpath on vertices of total actual weight at least $\lambda_2$. 
Thus any sequence of vertices of weight at most $\lambda_2$ that is contained in the path is contained in one of its overlapping subpaths. 
For each overlapping subpath, we shall maintain a succinct version of the corresponding sorted matrix $M(P)$ (as described in Section~\ref{sec:prelims}), where $P$ is the overlapping subpath excluding the top vertex of the overlapping subpath.

Initially, no path is heavy, since initially $\lambda_2=\infty$. 
Our algorithm recognizes a path $P$ as heavy in one of two ways.
\begin{enumerate}
\item
Path $P$ was middleweight until a feasibility test reduced the value of $\lambda_2$.
\item
Path $P$ results from the concatenation of two or more paths following the resolution of a problematic vertex.
\end{enumerate}

If a heavy path $P$ arises in the first way, then represent $P$ by two overlapping subpaths that are both copies of $P$, arbitrarily designating one as the first overlapping subpath and the other as the second. For each overlapping subpath, create the corresponding succinct description of a sorted matrix. 
If heavy path $P$ is the concatenation of paths, all of which were light, then do the same thing.

Otherwise, path $P$ results from the concatenation of paths, with at least one of them being heavy. 
Do the following to generate the representation for $P$. 
While there is a light or middleweight path $P'$ to be concatenated to a heavy path $P''$, combine the two as follows. If $P'$ precedes $P''$, then extend the first overlapping subpath of $P''$ to incorporate $P'$. 
If $P'$ follows $P''$, then extend the last overlapping subpath of $P''$. 
This completes the description of the concatenation of light or middleweight paths with heavy paths. 
While there is more than one heavy path, concatenate adjacent heavy paths $P'$ and $P''$ as follows. 
Assume that $P'$ precedes $P''$. Combine the last overlapping subpath of $P'$ with the first of $P''$. 
Note that any vertex can be changed at most twice before it is in overlapping subpaths that are neither the first nor the last. 

We now discuss how to perform efficient feasibility testing, given our representation of paths. 
To efficiently search along paths, we maintain a second representation of paths as red-black balanced search trees. 
In each tree, there will be a node $x$ for every vertex $v$ in the subpath. 
Node $x$ contains two fields, $wt(x)$ and $ct(x)$. Field $wt(x)$ contains the sum of the actual weights of all vertices whose nodes are in the subtree rooted at $x$, and field $ct(x)$ will equal the number of nodes in the subtree rooted at $x$. 
With this tree it is easy to search for a vertex $v$ in subpath $P$ such that $v$ is the first vertex in $P$ so that the sum of the actual weights to $v$ is at least a certain value, and to determine at the same time the position of $v$ in $P$. 
This search takes time proportional to the logarithm of the length of $P$. 
When two paths $P'$ and $P''$ need to be concatenated together, we merge the corresponding search trees in time proportional to the logarithm of the combined path length.

Moreover, a problematic vertex also needs access to the accumulated actual weight of each of its subtrees. 
Thus, we also maintain a linked list $child$, which points to all children of a vertex, if they exist. 
When leaf paths become resolved, remove the pointer from the parent vertex, but since we do not move the locations of the other descending paths until only one remains, we do not need to change the pointers for the other children. 
Thus, when we delete a leaf path, we update the subpaths and pointers accordingly, so that each subpath in the edge-path-partition remains in a contiguous block of memory.
\subsection{Analysis of our Tree Algorithm}
\label{sec:tree:analysis}
The analysis of our algorithm is not obvious because techniques in Section~\ref{sec:path} force the synthetic weight of data structure $\mathcal{M}$ to decrease monotonically as our algorithm progresses, but neither $\mathcal{U}$ nor $\mathcal{H}$ individually have this property. 
Moreover, neither $\mathcal{U}$ nor $\mathcal{H}$ alone provides an accurate count of the number of resolved paths. 

In a feasibility test, our algorithm must explore problematic vertices as well as both processed and pending paths. 
Specifically, the feasibility test time is the number of vertices our algorithm lands on in pending paths, plus the number of vertices our algorithm lands on in processed paths, plus the number of problematic vertices. 
Since leaf paths are pending paths, and the number of problematic vertices is less than the number of leaf paths, the feasibility test time is upper bounded by double the number of vertices our algorithm lands on in both pending and processed paths.
Thus, we show an upper bound on the time spent by the feasibility test after each round. 
Finally, we show that our algorithm needs at most $66(r^2+9)$ iterations to cut in half the feasibility test time.

We remark that cleaning and gluing will be the same as in Section~\ref{sec:path} since our algorithm performs cleaning and gluing only along subpaths that were initially created as part of the edge-path-partition or along heavy paths, for which we would have a corresponding submatrix. 

To analyze the performance of our algorithm, we consider, in an auxiliary data structure $\mathcal{P}$, the representatives of heavy and middleweight subpaths of the tree, along with their corresponding synthetic weights as a function of their lengths (as defined in Section~\ref{sec:prelims}). 
When we resolve a problematic vertex, the {\it handle\_leaves} routine produces a newly formed path that is either a light path, a heavy path, or a middleweight path. 
Recall that if a newly formed subpath is middleweight, we insert a representative consisting of actual weight of that subpath into $\mathcal{U}$. 
If a newly formed subpath is heavy, we generate our succinct representation of its corresponding submatrix, and we insert the representatives of the submatrix into $\mathcal{H}$. 
In each case, we also insert the representative(s) into $\mathcal{P}$ with their same synthetic weights, as described earlier. 
Furthermore, at the beginning of each iteration, $\mathcal{P}$ includes exactly $\mathcal{H}$, $\mathcal{U}$, and the representatives of middleweight subpaths of paths represented in $\mathcal{U}$. 
For convenience in analysis, let $n$ be the smallest power of two greater than the number of vertices in the graph. 
We now consider $\mathcal{P}$ to show results analogous to Lemma \ref{lem:3:2} and Lemma \ref{lem:3:3}.

\begin{lemma}
\label{lem:weight:bound}
At any point in our algorithm, the total synthetic weight of the values in $\mathcal{P}$ is less than $(4/3)*4n^5$.
\end{lemma}
\begin{proof}
Consider the synthetic weights assigned to the representative values in $\mathcal{P}$. 
Even if all representatives in $\mathcal{P}$ represent heavy paths, 
procedure {\it mats\_for\_path} creates
at most $n$ submatrices for subpaths of length $1$,
at most $n/2$ more submatrices for subpaths of length $1$,
at most $n/4$ submatrices for subpaths of length $2$, and so on,
up to at most one submatrix for a subpath of length $n/2$, 
corresponding to subpaths of lengths $1, 2, 4, \ldots , n$.
The total synthetic weight of the representative values corresponding to each heavy path length is at most
$4n^5, n^5, n^5/4, \ldots , 4n^3$, resp.
If a path is middleweight, no submatrix is generated for the path, but by construction, its 
representative value in $\mathcal{P}$ has the same synthetic weight as it would if the path were heavy.  
Thus, the total synthetic weight in $\mathcal{P}$ is less than $(4/3)*4n^5$.
\end{proof}

In Section~\ref{sec:tree} we defined functions $wgt$ and {\it eff\_wgt} on matrices to analyze the progress of our algorithm. 
We use a similar analysis in this section, but we do not have submatrices for middleweight paths. 
To emphasize the parallels, we overload the definitions of $wgt$ and {\it eff\_wgt}, as defined below. 
Let $wgt(P)$ be the synthetic weight assigned to subpath $P$, as a function of its length, as defined in Section~\ref{sec:prelims}. 
Define the {\it effective weight},
denoted {\it eff\_wgt}$(P)$, of a subpath $P$ containing a representative value in $\mathcal{P}$
to be $wgt(P)$ if both its smallest value and largest value
are contained in the interval $(\lambda_1, \lambda_2)$
and $(3/4)wgt(P)$ if only one of its smallest value and largest value
is contained in the interval $(\lambda_1, \lambda_2)$. 
We analyze our algorithm using this notion of {\it eff\_wgt}$(P)$ without ever the need to explicitly calculate it. 
Let {\it eff\_wgt}$(\mathcal{P})$ be the total effective weight of all subpaths with representative values in $\mathcal{P}$.

\begin{lemma}
\label{lem:selection:reduction}
A weighted and unweighted selection, first both from $\mathcal{U}$ and then both from $\mathcal{H}$, resolves at least $1/24$ of {\it eff\_wgt}$(\mathcal{P})$
\end{lemma}
\begin{proof}
When a feasibility test renders a value (the smallest or largest)
from a heavy subpath $P$ no longer in the interval $(\lambda_1, \lambda_2)$,
we argue that {\it eff\_wgt}$(P)$ is reduced by at least $wgt(P)/4$
because of that value.

If both values were contained in the interval $(\lambda_1, \lambda_2)$,
but one no longer is, then clearly {\it eff\_wgt}$(P)$ is reduced by at least $wgt(P)/4$.
If only one value was contained in the interval $(\lambda_1, \lambda_2)$,
and it no longer is, then we consider whether $\lambda_1$ increased or $\lambda_2$ decreased. 
If $\lambda_1$ increased and $P$ has no value within $(\lambda_1, \lambda_2)$, then any subpath will also 
have no value within $(\lambda_1, \lambda_2)$. 
Thus, {\it eff\_wgt}$(P)$ is clearly reduced by at least $wgt(P)/4$. 
On the other hand, if $\lambda_2$ decreased so that $P$ is now a heavy path, then we generate $M(P)$. 
We then replace $M(P)$ with the succinct description of four submatrices of effective weights $wgt(P)/8 + wgt(P)/8 + (3/4)wgt(P)/8 + (3/4)wgt(P)/8$ $< wgt(P)/2$, so that there is a reduction in weight by greater than $wgt(P)/4$.
If both values from the submatrix were contained in the interval $(\lambda_1, \lambda_2)$,
and both no longer are, then we consider first one and then the other.

Thus every representative that was in $(\lambda_1, \lambda_2)$ but no longer is causes a decrease in {\it eff\_wgt}$(\mathcal{P})$ by an amount at least equal to its synthetic weight in $\mathcal{P}$.
Representatives with more than half of the total weight in $\mathcal{P}$ find themselves no longer in $(\lambda_1, \lambda_2)$.

Recall that we include values representing a subpath $P$ in selection set $\mathcal{H}$ if and only if $P$ is a heavy subpath, and the corresponding values are contained within $(\lambda_1, \lambda_2)$. 
Of a subpath $P$ that does have representative values contained in $\mathcal{H}$, then $P$ has either two values in $\mathcal{P}$ at a total of $2wgt(P)/4 = (1/2)wgt(P)$ or one value at a weight of $wgt(P)/4 = (1/3)(3wgt(P)/4)$. 
For each selection and feasibility test from $\mathcal{H}$ and $\mathcal{U}$, at least half of the weight of $\mathcal{P}$ is contained in $\mathcal{U}$ or contained in $\mathcal{H}$. 

Suppose at least $1/2$ of $\mathcal{P}$ is contained in $\mathcal{U}$. 
Then following a selection and feasibility test from $\mathcal{U}$, at least $(1/2)(1/2)=1/4$ of the weight of $\mathcal{P}$ will be removed from $\mathcal{U}$, either determined to be smaller than an updated $\lambda_1$, or larger than an updated $\lambda_2$ and inserted into $\mathcal{H}$ as a part of a heavy subpath. 
Then, following a round of selection and feasibility testing from $\mathcal{H}$, the weight of $\mathcal{P}$ is reduced by at least $(1/2)(1/4)=1/8$.

On the other hand, if there is more weight of $\mathcal{P}$ in $\mathcal{H}$ than in $\mathcal{U}$, then a round of selection and feasibility testing 
from $\mathcal{H}$ decreases the weight of $\mathcal{P}$ by at least $(1/2)(1/2)=1/4$.

Thus, in the worse of the two cases, the weight of $\mathcal{P}$ decreases by at least $(1/8)(1/3) = 1/24$ per iteration.
\end{proof}

\begin{lemma}
\label{lem:resolved:paths}
Following $i$ iterations of weighted and unweighted selection and feasibility testing for $\mathcal{H}$, where $2^{5j}=(3/4)*(24/23)^i$, at most $1/2^{5k}$ of the subpaths, of length $2^{j-k}$, represented by values in $\mathcal{H}$ can be pending. (Recall that a pending path is defined to contain some value that is not resolved.)
\end{lemma}
We omit the proof of Lemma~\ref{lem:resolved:paths}, as it is essentially contained in the proof of Lemma~\ref{lem:3:2}, with $24/23$ replacing $6/5$.

\begin{lemma}
\label{lem:heavy:submatrices}
The total time to generate and maintain the overlapping subpaths and thus, the representatives of the corresponding submatrices for all the heavy subpaths is $O(n)$.
\end{lemma}
\begin{proof}
Recall that each vertex can be in at most two overlapping subpaths, so the total time for generating the representatives of the corresponding submatrices for all heavy subpaths is at most $2n$. 
\end{proof}

\begin{theorem}
\label{thm:selection}
The time for inserting, selecting, and deleting representatives of $\mathcal{H}$ and $\mathcal{U}$ across all iterations is $O(n)$.
\end{theorem}
\begin{proof}
First, we count the number of representatives inserted into $\mathcal{H}$.
Initially, at most $2n-1$ subpaths have representatives in $\mathcal{H}$. 
For $j = 1, 2, \ldots,\log n - 1$, consider all subpaths of length $2^j$ that at some point have representatives in $\mathcal{H}$. 
A path that can be split must have its smallest value at most $\lambda_2$ and its largest value at least $\lambda_1$. 
Note that light paths are not split further, while heavy paths are represented succinctly by submatrices. 
However, $M_{i,j}>M_{i-k,j+k}$ for $k>0$, since the heavy path represented by $M_{i-k,j+k}$ is a subpath of the path represented by $M_{i,j}$. 
Hence, for any submatrix of size $2^j\times 2^j$ which is split, at most one submatrix can be split in each diagonal extending upwards from left to right. 
There are fewer than $2n$ diagonals, so there will be fewer than $2(n/2^j)$ submatrices that are split. 
Thus the number of submatrices resulting from quartering is less than $8(n/2^j)$. 
Summing over all $j$ gives $O(n)$ insertions into $\mathcal{H}$ overall.

Next, we count the number of representatives inserted into $\mathcal{U}$. 
There are at most $n$ subpaths of length $1$, $n/2$ subpaths of length $2$, and so forth, up to at most $1$ path of length $n$. 
Since we insert a representative of each subpath at most once into $\mathcal{U}$, the number of representatives inserted into $\mathcal{U}$ is $O(n)$.

Finally, we show that the total selection time from $\mathcal{H}$ and $\mathcal{U}$ is linear in $n$. 
We give an accounting argument similar to that used in the proof of Lemma~\ref{lem:3:3}. 
Charge $2$ credits for each value inserted into $\mathcal{H}$ or $\mathcal{U}$. 
As $\mathcal{H}$ and $\mathcal{U}$ change, we maintain the invariant that the number of credits is twice the size of $\mathcal{H}$ or $\mathcal{U}$, respectively. 
The rest of the proof is analogous to that of Lemma~\ref{lem:3:3}.
Thus, we conclude that the time for performing selection is $O(n)$.
\end{proof}

\noindent
We make the following observations: The time spent by the feasibility test on a pending path is at least as much time as spent by the feasibility test as if the pending path were a processed path. Resolving a processed path does not increase the feasibility test time spent on pending paths. 
Based on these observations, we note that the feasibility test time spent on pending paths cannot increase and once all pending paths are resolved, the time spent on processed paths cannot increase.
We analyze our algorithm in each of the following three exhaustive cases:
\begin{enumerate}
\item
Lemma~\ref{lem:pending:paths}: The feasibility test spends more time on heavy paths than middleweight paths and at least as much time on pending paths as on processed paths.
\item
Lemma~\ref{lem:processed:paths}: The feasibility test spends more time on heavy paths than middleweight paths and more time on processed paths than pending paths.
\item
Lemma~\ref{lem:m:r}: The feasibility test spends at least as much time on middleweight paths as on heavy paths.
\end{enumerate}

\begin{lemma}
\label{lem:pending:paths}
Suppose at the beginning of round $r$, that the feasibility test lands on at most $n/2^{r-1}$ vertices. Suppose further that the feasibility test lands on more vertices in heavy paths than middleweight paths. If the feasibility test also lands on at least as many vertices in pending paths as vertices in processed paths, then following at most $163r+170$ iterations of selection and feasibility testing from $\mathcal{H}$, the number of vertices in pending paths that the feasibility test lands on is either halved or at most $n/2^{r+2}$.
\end{lemma}
\begin{proof}
By assumption, the feasibility test lands on at most $n/2^{r-1}$ vertices and lands on more vertices in heavy paths than middleweight paths, and at least as many vertices in pending paths as vertices in processed paths. 
If the number of vertices in pending paths that the feasibility test lands on is at most $n/2^{r+2}$, then the result follows. 
Thus, we assume the feasibility test lands on more than $n/2^{r+2}$ vertices in pending paths. 
By Lemma \ref{lem:resolved:paths}, in iteration $i$, at most $1/2^{5k}$ of the subpaths of length $2^{j-k}$ can be pending, where $2^{5j}=(3/4)*(24/23)^i$. 
Then for $j=2r+2$ and $k=r$, at most $1/2^{5r}$ of the subpaths of length $2^{(r+2)}$ can be pending, so there are at most $n/2^{5r}$ vertices remaining in pending paths. 
Thus, it takes at most $i=(5(2r+2)-\log(3/4))/\log(24/23)<163r+170$ iterations of selection and feasibility testing from $\mathcal{H}$ to reduce the amount of time spent on the pending paths by at least half. 
Since each iteration of feasibility testing lands on at most $n/2^{r-1}$ vertices, the total number of vertices checked is at most $\left(n/2^{r-1}\right)(163r+170)$.
\end{proof}

\noindent
Before we can show an analogous result for processed paths, we introduce three preliminary lemmas.

\begin{lemma}
\label{lem:processed:time}
If a feasible lands on at least $t$ vertices in a processed path $P$, then $P$ has length at least $2^{\sqrt{t/2}}$.
\end{lemma}
\begin{proof}
Suppose $P$ has length less than $2^{\sqrt{t/2}}$. 
From the edge-path-partition, $P$ can contain disjoint subpaths of lengths $1,2,2^2,\ldots,2^{\sqrt{t/2}-1}$. 
However, if $P$ previously contained a problematic vertex, then it may have two disjoint subpaths of each length. 
Thus, the feasibility test lands on at most $2(1+2+\ldots+\sqrt{t/2}-1)=(\sqrt{t/2}-1)\sqrt{t/2}<t$ vertices in total in $P$, which is a contradiction.
\end{proof}

\begin{lemma}
\label{lem:median:leaf}
Suppose a feasibility test lands on at most $n/2^r$ vertices but more than $n/2^{r+2}$ vertices. If the feasibility test spends more time on processed paths than pending paths, then the median length of leaf paths is at most $\max(2^{4r^2},2^{400})$.
\end{lemma}
\begin{proof}
Suppose, by way of contradiction, the median length of a leaf path is more than $2^{4r^2}$ and $r>10$. 
Let $x$ be the mean number of vertices the feasibility test lands on, across all leaf paths, so that $4r^2<x$. 
Let $t$ be the mean number of vertices the feasibility test lands on, across all processed paths. 
Recall that all leaf paths are pending paths. 
By assumption, the feasibility test spends more time on processed paths than pending paths. 
Moreover, the number of leaf paths is more than the number of internal paths, pending or processed, so that $t>x$. 
By Lemma~\ref{lem:processed:time}, each processed path in which the feasibility test lands on at least $t$ vertices has length is at least $2^{\sqrt{t/2}}$. 
Thus, if $x\le 2^r$, then by considering the pending paths, the ratio of the time spent by the feasibility test to the total number of vertices is at most $\frac{x}{2^{4r^2}}<\frac{1}{2^{r+2}}$, which contradicts the assumption that the feasibility test lands on more than $n/2^{r+2}$ vertices.
On the other hand, if $x>2^r$, then by considering the processed paths, the ratio of the time spent by the feasibility test to the total number of vertices is at most $\frac{x+t}{2^{4r^2}+2^{\sqrt{t/2}}}<\frac{2t}{2^{\sqrt{t/2}}}<\frac{1}{2^{r+2}}$ for $r\ge10$, since $t>x>2^r$. 
This again contradicts the assumption that the feasibility test lands on more than $n/2^{r+2}$ vertices.
Thus, the median length of a leaf path is at most $\max(2^{4r^2},2^{400})$.
\end{proof}

\noindent
For the remainder of the section, we analyze $r\ge10$, noting that for $r<10$, the median length of a leaf path is at most $2^{400}$ and can be handled in a constant number of feasibility tests.
\begin{lemma}
\label{lem:median:processed}
Suppose the feasibility test lands on at most $n/2^r$ vertices but more than $n/2^{r+1}$ vertices. Suppose further that the feasibility test lands on more vertices in heavy paths than middleweight paths. If the feasibility test spends more time on processed paths than pending paths, then the median length of processed paths is at most $2^{r^2+9}$. Hence, the number of vertices the feasibility test lands on in a median length processed path is at most $2(r^2+9)^2$.
\end{lemma}
\begin{proof}
Suppose, by way of contradiction, the median length of processed paths is more than $2^{r^2+9}$. 
Then by Lemma \ref{lem:processed:time}, the number of vertices in processed paths that the feasibility test lands on is at most $\left(2(r^2+9)^2/2^{r^2+9}\right)n$. 
By assumption, the feasibility test spends more time on heavy paths than middleweight paths, and more time on processed paths than pending paths, so the number of vertices in processed paths that the feasibility test lands on is at least $n/2^{r+2}$. 
But for all positive integers $i$, it holds that $1/2^{i+2}>2(i^2+9)/2^{i^2+9}$. Thus, $n/2^{r+2}>\left(2(r^2+9)^2/2^{r^2+9}\right)n$, which contradicts the assumption that the feasibility test lands on more than $n/2^{r+1}$ vertices. 
Hence, the median length of a processed path is at most $2^{r^2+9}$.
\end{proof}

\noindent
We are now ready to show the reduction in processed paths from repeated instances of feasibility testing.

\begin{lemma}
\label{lem:processed:paths}
Suppose at the beginning of round $r$, that the feasibility test lands on at most $n/2^{r-1}$ vertices. 
Suppose further that the feasibility test lands on more vertices in heavy paths than middleweight paths. 
If the feasibility test lands on more vertices in processed paths than vertices in pending paths, then following at most $6(r^2+9)(408r^2+815r+415)$ iterations of selection and feasibility testing from $\mathcal{H}$, which takes $O(nr^4/2^r)$ time, the number of vertices in processed paths that the feasibility test lands on is either halved or at most $n/2^{r+2}$.
\end{lemma}
\begin{proof}
By assumption, the feasibility test lands on at most $n/2^{r-1}$ vertices and lands on more vertices in heavy paths than middleweight paths, and more vertices in processed paths than vertices in pending paths. 
If the number of vertices in processed paths that the feasibility test lands on is at most $n/2^{r+2}$, then the lemma immediately follows. 
Thus, we assume the feasibility test lands on more than $n/2^{r+2}$ vertices in processed paths.

By Lemma \ref{lem:median:leaf}, the median length of a leaf path is at most $2^{4(r-1)^2}$. 
By Lemma \ref{lem:resolved:paths}, in iteration $i$, at most $1/2^{5k}$ of the subpaths of length $2^{j-k}$ can be pending, where $2^{5j}=(3/4)(24/23)^i$. 
Then for $j=5(r+1)^2$ and $k=(r+1)^2$, at most $1/2^{5(r+1)^2}$ of the subpaths of length $2^{4(r+1)^2}>2^{4(r-1)^2}$ can be pending, so at least half of the leaf paths are resolved. 
It takes at most $i=(25(r+1)^2-\log(3/4))/\log(24/23)<408r^2+815r+414$ iterations to resolve half of the leaf paths, so that the appropriate values can be inserted into $\mathcal{V}$. 
One more iteration of feasibility testing is run using the selected median from $\mathcal{V}$. 
Thus, running at most $408r^2+815r+414$ iterations of {\it handle\_pending\_paths}, followed by an iteration of {\it handle\_processed\_paths}  (for a total of at most $408r^2+815r+415$ iterations) reduces the number of leaf paths by a factor of $1/4$, or equivalently, reduces the total number of paths by a factor of $1/8$.

Since $(7/8)^6<1/2$, then by repeating at most $6(r^2+9)$ times, the total number of paths is reduced by a factor of at least $1/2^{r^2+9}$. 
If the average length of the remaining processed paths is more than $2^{r^2+9}$, then by Lemma \ref{lem:median:processed}, the feasibility test lands on at most $n/2^{r+1}$ vertices, which is a reduction of $1/2>1/4$ in the number of vertices checked by the feasibility test. 
Otherwise, if the average length of the remaining processed paths is less than $2^{r^2+9}$, then by reducing the total number of paths by factor of at least $1/2^{r^2+9}$, the time spent by the feasibility test on vertices in processed paths is at least halved. 
We require at most $6(r^2+9)$ cycles, each with $408r^2+815r+415$ iterations of feasibility testing. 
Thus, at most $6(r^2+9)(408r^2+815r+415)=O(r^4)$ iterations are needed to reduce the feasibility test by at least half, each checking at most $n/2^r$ vertices, for a total of $O(nr^4/2^r)$ time.
\end{proof}

\begin{lemma}
\label{lem:m:r}
For each round, the feasibility test time is reduced by at least $1/2$. 
Thus at the beginning of round $r$, the feasibility test lands on at most $n/2^{r-1}$ vertices.
\end{lemma}
\begin{proof}
We note that prior to the first round, the feasibility test lands on exactly $n$ vertices, and we proceed via induction. 
Suppose at the beginning of round $r$, the feasibility test lands on at most $n/2^{r-1}$ vertices. 
If the feasibility test in fact lands on at most $n/2^r$ vertices, then the induction already holds. 
Recall that we have three cases, as we claimed just before Lemma~\ref{lem:pending:paths}. 
The feasibility test can spend at least as much time on middleweight paths as on heavy paths. 
Otherwise, the feasibility test spends more time on heavy paths than middleweight paths, but can spend more time on either pending paths or processed paths.  

If the feasibility test spends at least as much time on pending paths as on processed paths, then by Lemma \ref{lem:pending:paths}, no more than $163r+170$, which is certainly less than $408r^2+815r+415$, iterations of selection and feasibility from $\mathcal{H}$ are needed to reduce the portion spent by the feasibility test on pending paths by at least half. 
Hence, the overall feasibility test time is reduced by at least $1/8$.

Otherwise, if the feasibility test spends more time on processed paths than pending paths, then by Lemma \ref{lem:processed:paths}, at most $(6r^2+9)(408r^2+815r+415)$ iterations of selection and feasibility from $\mathcal{H}$ are needed to reduce the portion spent by the feasibility test on processed paths by at least half. 
Hence, the overall feasibility test time is reduced by at least $1/8$.

If the feasibility test spends at least as much time on middleweight paths as on heavy paths, then following a selection in $\mathcal{U}$ that is weighted by path length, and then a feasibility test, at least half the vertices in middleweight paths will be determined to be either in light paths or in heavy paths. 
If the vertices are determined to be in light paths, then the portion spent by the feasibility test on these paths is reduced by at least half, since newly formed light paths are cleaned and glued in \emph{handle\_middleweight\_paths}, without increasing the time spent by the feasibility test on heavy paths. 
Hence, if $\lambda_1$ is increased by a feasibility test from $\mathcal{U}$, the overall feasibility test time is reduced by at least $1/8$. 

However, if the vertices are determined to be in heavy paths, our algorithm will insert the representatives of the corresponding submatrices into $\mathcal{H}$. 
Our algorithm thus reduces the portion spent by the feasibility test on heavy paths by at least $1/8$.
Hence, taking the reductions of both types into account, the overall feasibility test time is reduced by at least $1/16$.
Thus, we can reduce the overall number of vertices that the feasibility test lands on by $1/16$ in $O(nr^4/2^r)$ time, for each $r$. 
Since $(15/16)^{11}<1/2$, then at most eleven repetitions suffice to halve the overall number of vertices checked by the feasibility test. 
Indeed, each round requires at most $11[6(r^2+9)]=66(r^2+9)$ repetitions, and so at the beginning of round $r+1$, the runtime is at most $n/2^r$.
\end{proof}

\begin{corollary}
\label{cor:total:calls}
Round $r$ has $O(r^4)$ calls to the feasibility test. 
\end{corollary}
\begin{proof}
Since round $r$ requires at most $66(r^2+9)$ repetitions of the inner loop, and the inner loop uses at most $(408r^2+815r+415)$ feasibility tests, then the total number of feasibility tests in round $r$ is $O(r^4)$.
\end{proof}

\noindent
Now, we claim the main result of paper.

\begin{theorem}
The runtime of our algorithm is $O(n)$.
\end{theorem}
\begin{proof}
By Lemma \ref{lem:m:r}, the feasibility test lands on at most $n/2^{r-1}$ vertices at the beginning of round $r$. 
By Corollary \ref{cor:total:calls}, round $r$ has $O(r^4)$ calls to the feasibility test, which take $O(n/2^{r-1})$ time each.
Hence, there are at most $\log n$ rounds. 
Note that $\sum_{r=0}^{\log n} nr^4/2^{r-1}$ is $O(n)$. By Theorem \ref{thm:selection}, the total time for handling $\mathcal{H}$ and $\mathcal{U}$ and performing selection over all iterations is $O(n)$, while the same result clearly holds for $\mathcal{V}$. 
Therefore, the total time required by our algorithm is $O(n)$.
\end{proof}
\section{Conclusion}
Our algorithms solve the max-min and min-max path $k$-partition problems, as well as the max-min tree $k$-partition problem in asymptotically optimal time. 
Consequently, they use the paradigm of parametric search without incurring any $\log$ factors. 
We avoid these $\log$ factors by searching within collections of candidate values that are implicitly encoded as matrices or as single values. 
We assign synthetic weights for these values based on their corresponding path lengths so that weighted selections favor the early resolution of shorter paths, which will speed up subsequent feasibility tests. 
In fact, our analysis relies on demonstrating a constant fraction reduction in the feasibility test time as the algorithm progresses. 
Simultaneously, unweighted selections quickly reduce the overall size of the set of candidate values so that selecting test values is also achieved in linear time. 

Furthermore, we have successfully addressed the challenge of developing a meaningful quantity to track progress as the algorithm proceeds. 
We have proved that the time to perform a feasibility test describes in a natural way the overall progress from the beginning of the algorithm. 
Unfortunately, even these observations alone are not enough to overcome the challenges of tree partitioning. 
In particular, without both compressing long paths and pruning leaf paths quickly, the feasibility test time might improve too slowly for the overall algorithm to take linear time. 
Our dual-pronged strategy addresses both issues simultaneously, demonstrating that parallel algorithms are not essential or even helpful in designing optimal algorithms for certain parametric search problems.

\baselineskip 18pt
\bibliographystyle{plain}
\bibliography{all}
\end{document}